\documentclass[sigconf]{aamas}  
\usepackage{balance}  

\usepackage{booktabs}

\setcopyright{ifaamas}  
\acmDOI{}  
\acmISBN{}  
\acmConference[AAMAS'19]{Proc.\@ of the 18th International Conference on Autonomous Agents and Multiagent Systems (AAMAS 2019)}{May 13--17, 2019}{Montreal, Canada}{N.~Agmon, M.~E.~Taylor, E.~Elkind, M.~Veloso (eds.)}  
\acmYear{2019}  
\copyrightyear{2019}  
\acmPrice{}  

\usepackage{amsmath,amsfonts,amssymb,amsthm}
\usepackage{graphicx}
\usepackage{subcaption}
\usepackage{acro}

\usepackage[]{algorithm2e}
\usepackage{lipsum}
\usepackage{enumitem}
\usepackage{thm-restate}	

\usepackage{tikz} 
	\usetikzlibrary{shapes} 
	\usetikzlibrary{calc} 
\newlength{\nodesize}
\colorlet{chance_color}{black}
\colorlet{pl0_color}{chance_color}
\colorlet{chance_text}{white}
\colorlet{pl1_color}{magenta!50}
\colorlet{pl2_color}{green!50!lime!60}
\tikzset{
	basenode/.style = {draw,
		inner sep = 0.1em,
		minimum size = \nodesize
	},
	playernode/.style={basenode,
		shape = regular polygon,
		regular polygon sides = 3
	},
	pl1/.style={playernode, fill=pl1_color},
	pl2/.style={playernode, fill=pl2_color, shape border rotate=180},
	chance/.style = {basenode,
		fill=pl0_color, text=chance_text,
		circle,
		minimum size=0.7*\nodesize,
	},
	terminal/.style = {basenode,
		draw=none,
		outer sep=0,
		minimum size = 0.6\nodesize
	}
}

\setlipsumdefault{0}

\graphicspath{ {../images/}, {./images/} }

\newtheorem{remark}[theorem]{\protect\remarkname}

\newtheorem*{lemma*}{\protect\lemmaname}
\newtheorem*{definition*}{\protect\definitionname}

\providecommand{\claimname}{Claim}
\providecommand{\definitionname}{Definition}
\providecommand{\lemmaname}{Lemma}
\providecommand{\notationname}{Notation}
\providecommand{\remarkname}{Remark}
\providecommand{\problemname}{Problem}

\newcommand{\R}{\mathbb{R}}

\newcommand{\mc}{\mathcal}

\newcommand{\expl}{\textnormal{expl}} 
\newcommand{\opponent}{{\textnormal{opp}_i}} 
\newcommand{\aug}{\mc I^\textnormal{aug}} 
\newcommand{\dbbar}[1]{\bar{\bar{#1}}}

\newcounter{vkNoteCounter}
\newcommand{\vknote}[1]{\textnormal{{\scriptsize \color{blue} $\blacktriangle$ \refstepcounter{vkNoteCounter}\textsf{[VK]$_{\arabic{vkNoteCounter}}$: {#1}}}}}
\renewcommand{\vknote}[1]{}

\newcounter{msNoteCounter}
\newcommand{\msnote}[1]{\textnormal{{\scriptsize \color{green} $\blacksquare$ \refstepcounter{msNoteCounter}\textsf{[MS]$_{\arabic{msNoteCounter}}$: {#1}}}}}
\renewcommand{\msnote}[1]{}

\newcounter{vlNoteCounter}

\renewcommand{\msnote}[1]{}

\definecolor{insignwin}{rgb}{0.5,0.5,0.5}
\definecolor{insignlose}{rgb}{0.5,0.5,0.5}
\definecolor{win}{rgb}{0,0,0}
\definecolor{lose}{rgb}{0,0,0}

\DeclareAcronym{CFV}{
short = CFV ,
long  = CounterFactual Value $v$,
class = abbrev
}
\DeclareAcronym{IIG}{
 short = IIG,
 long = Imperfect-Information Game,
 class=abbrev
}
\DeclareAcronym{MCCR}{
 short = MCCR,
 long = {M}onte {C}arlo Continual Resolving,
 class=abbrev
}
\DeclareAcronym{NE}{
 short = NE,
 long = {N}ash Equilibrium,
 class=abbrev
}
\DeclareAcronym{EFG}{
 short = EFG,
 long = Extensive Form Game,
 class=abbrev
}
\DeclareAcronym{OS}{
 short = OS,
 long = Outcome Sampling,
 class=abbrev
}
\DeclareAcronym{OOS}{
 short = OOS,
 long = Online Outcome Sampling,
 class=abbrev
}
\DeclareAcronym{OOS-IST}{
 short = OOS-IST,
 long = OOS with Information Set Targeting,
 class=abbrev
}
\DeclareAcronym{OOS-PST}{
 short = OOS-PST,
 long = OOS with Public State Targeting,
 class=abbrev
}
\DeclareAcronym{IS}{
 short = IS,
 long = Information Set,
 class=abbrev
}
\DeclareAcronym{CFR}{
 short = CFR,
 long = CounterFactual Regret minimization,
 class=abbrev
}
\DeclareAcronym{MCCFR}{
 short = MCCFR,
 long = Monte Carlo CounterFactual Regret minimization,
 class=abbrev
}
\DeclareAcronym{IS-MCTS}{
 short = IS-MCTS,
 long = Information Set Monte Carlo Tree Search,
 class=abbrev
}
\DeclareAcronym{MCTS}{
 short = MCTS,
 long = Monte Carlo Tree Search,
 class=abbrev
}
\DeclareAcronym{UCT}{
 short = UCT,
 long = Upper Confidence bound applied to Trees,
 class=abbrev
}
\DeclareAcronym{RM}{
 short = RM,
 long = Regret Matching,
 class=abbrev
}

\begin{document}

\title[Monte Carlo Continual Resolving in Imperfect Information Games]{Monte Carlo Continual Resolving for Online Strategy Computation in Imperfect Information Games}  





\author{Michal \v Sustr}
\orcid{0000-0002-3154-4727}
\affiliation{}
\email{michal.sustr@aic.fel.cvut.cz}

\author{Vojt\v ech Kova\v r\'ik}
\orcid{0000-0002-7954-9420}
\email{vojta.kovarik@gmail.com}
\affiliation{%
    \vspace{3mm}
    \institution{Artificial Intelligence Center, FEE Czech Technical University}
    \city{Prague}
    \state{Czech Republic}
    \vspace{3mm}
}

\author{Viliam Lis\'y}
\orcid{0000-0002-1647-1507}
\email{viliam.lisy@fel.cvut.cz}
\affiliation{}


\begin{abstract}  

Online game playing algorithms produce high-quality strategies with a fraction of memory and computation required by their offline alternatives.
Continual Resolving (CR) is a recent theoretically sound approach to online game playing that has been used to outperform human professionals in poker.
However, parts of the algorithm were specific to poker, which enjoys many properties not shared by other imperfect information games.
We present a domain-independent formulation of CR applicable to any two-player zero-sum extensive-form games (EFGs).
It works with an abstract resolving algorithm, which can be instantiated by various EFG solvers.
We further describe and implement its Monte Carlo variant (\acs{MCCR}) which uses Monte Carlo Counterfactual Regret Minimization (\acs{MCCFR}) as a resolver.
We prove the correctness of CR and show an $O(T^{-1/2})$-dependence of \acs{MCCR}'s exploitability on the computation time.
Furthermore, we present an empirical comparison of \acs{MCCR} with incremental tree building to Online Outcome Sampling and Information-set \acs{MCTS} on several domains.


\end{abstract}

%


\maketitle


\section{Introduction}

Strategies for playing games can be pre-computed \emph{offline} for all possible situations, or computed \emph{online} only for the situations that occur in a particular match. The advantage of the offline computation are stronger bounds on the quality of the computed strategy. Therefore, it is preferable if we want to solve a game optimally. On the other hand, online algorithms can produce strong strategies with a~fraction of memory and time requirements of the offline approaches. Online game playing algorithms have outperformed humans in Chess~\cite{DeepBlue}, Go~\cite{Silver16:AlphaGo}, and no-limit Poker~\cite{brown2017safe,DeepStack}.

While online approaches have always been the method of choice for strong play in perfect information games, it is less clear how to apply them in imperfect information games (\acs{IIG}s). To find the optimal strategy for a specific situation in an \ac{IIG}, a player has to reason about the unknown parts of the game state. They depend on the (possibly unobservable) actions of the opponent prior to the situation, which in turn depends on what the optimal decisions are for both players in many other parts of the game. This makes the optimal strategies in distinct parts of the game closely interdependent and makes correct reasoning about the current situation difficult without solving the game as a whole.

Existing online game playing algorithms for imperfect information games either do not provide any guarantees on the quality of the strategy they produce~\cite{long2010understanding,Ciancarini2010,cowling2012information}, or require the existence of a compact heuristic evaluation function and a significant amount of computation to construct it~\cite{DeepStack,brown2018depth}. Moreover, the algorithms that are theoretically sound were developed primarily for Texas hold'em poker, which has a very particular information structure. After the initial cards are dealt, all of the actions and chance outcomes that follow are perfectly observable. Furthermore, since the players' moves alternate, the number of actions taken by each player is always known.
None of this holds in general for games that can be represented as two-player zero-sum extensive-form games (\acs{EFG}s). In a blind chess~\cite{Ciancarini2010}, we may learn we have lost a piece, but not necessarily which of the opponent's pieces took it. In visibility-based pursuit-evasion~\cite{raboin2010strategy}, we may know the opponent remained hidden, but not in which direction she moved. In phantom games~\cite{teytaud2011lemmas}, we may learn it is our turn to play, but not how many illegal moves has the opponent attempted. Because of these complications, the previous theoretically sound algorithms for imperfect-information games are no longer directly applicable.

The sole exception is Online Outcome Sampling (\acs{OOS})~\cite{OOS}. It is theoretically sound, completely domain independent, and it does not use any pre-computed evaluation function. However, it starts all its samples from the beginning of the game, and it has to keep sampling actions that cannot occur in the match anymore. As a result, its memory requirements grow as more and more actions are taken in the match, and the high variance in its importance sampling corrections slows down the convergence.

We revisit the Continual Resolving algorithm (CR) introduced in~\cite{DeepStack} for poker and show how it can be generalized in a way that can handle the complications of general two-player zero-sum \ac{EFG}s.
Based on this generic algorithm, we introduce Monte Carlo Continual Resolving (\ac{MCCR}), which combines \ac{MCCFR}~\cite{MCCFR} with incremental construction of the game tree, similarly to \ac{OOS}, but replaces its targeted sampling scheme by Continual Resolving. This leads to faster sampling since \ac{MCCR} starts its samples not from the root, but from the current point in the game. It also decreases the memory requirements by not having to maintain statistics about parts of the game no longer relevant to the current match.
Furthermore, it allows evaluating continual resolving in various domains, without the need to construct expensive evaluation functions.

We prove that \ac{MCCR}'s exploitability approaches zero with increasing computational resources and verify this property empirically in multiple domains.
We present an extensive experimental comparison of \ac{MCCR} with \ac{OOS}, Information-set Monte Carlo Tree Search (\acs{IS-MCTS})~\cite{cowling2012information} and \ac{MCCFR}.
We show that \ac{MCCR}'s performance heavily depends on its ability to quickly estimate key statistics close to the root, which is good in some domains, but insufficient in others.

%
%
%
%
%
%

\section{Background}

We now describe the standard notation for \ac{IIG}s and \ac{MCCFR}.

\subsection{Imperfect Information Games}

We focus on two-player zero-sum extensive-form games with imperfect information. 
Based on~\cite{osborne1994course}, game $G$ can be described by
\begin{itemize}
	\item $\mc{H}$ -- the set of \emph{histories}, representing sequences of actions.
	\item $\mc{Z}$ -- the set of terminal histories (those $z\in \mc H$ which are not a prefix of any other history). We use $g \sqsubset h$ to denote the fact that $g$ is equal to or a prefix of $h$.
	\item $\mc A(h) := \{ a \, | \ ha \in \mc H \}$ denotes the set of actions available at a \emph{non-terminal history} $h\in \mc H\setminus \mc Z$. The term $ha$ refers to a history, i.e. child of history $h$ by playing $a$.
	\item $\mc P : \mc H \setminus \mc Z \rightarrow \{1,2,c\}$ is the \emph{player function} partitioning non-terminal histories into $\mc H_1$, $\mc H_2$ and $\mc H_c$ depending on which player chooses an action at $h$. Player $c$ is a special player, called ``chance'' or ``nature''.
    \item \emph{The strategy of chance} is a fixed probability distribution $\sigma_c$ over actions in chance player's histories.
    \item The \emph{utility function} $u=(u_1,u_2)$ assigns to each terminal history $z$ the rewards $u_1(z), u_2(z)\in \R$ received by players 1 and 2 upon reaching $z$. We assume that $u_2 = - u_1$.
	\item The \emph{information-partition} $\mc I = (\mc I_1, \mc I_2)$ captures the imperfect information of $G$. For each player $i\in \{1,2\}$, $\mc I_i$ is a partition of $\mc H_i$. If $g,h\in \mc H_i$ belong to the same $I\in \mc I_i$ then $i$ cannot distinguish between them. Actions available at infoset $I$ are the same as in each history $h$ of $I$, therefore we overload $\mc A(I) := \mc A(h)$. We only consider games with \emph{perfect recall}, where the players always remember their past actions and the information sets visited so far.
\end{itemize}

A \emph{behavioral strategy} $\sigma_i \in \Sigma_i$ of player $i$ assigns to each $I\in \mc I_i$ a probability distribution $\sigma(I)$ over available actions $a\in \mc A(I)$.
A \emph{strategy profile} (or simply \emph{strategy}) $\sigma = (\sigma_1,\sigma_2) \in \Sigma_1 \times \Sigma_2$ consists of strategies of players~1~and~2. For a~player $i \in \{1,2\}$, $-i$ will be used to denote the other two actors $\{1,2,c\}\setminus \{i\}$ in $G$ (for example $\mc H_{-1} := \mc H_{2} \cup \mc H_c$) and $\opponent$ denotes $i$'s opponent ($\textrm{opp}_1 := 2$).

\subsection{Nash Equilibria and Counterfactual Values}

The~\emph{reach probability} of a history $h\in \mc H$ under $\sigma$ is defined as $\pi^{\sigma}(h)=\pi^{\sigma}_{1}(h)\pi^{\sigma}_{2}(h)\pi^\sigma_c(h)$, where each $\pi^{\sigma}_{i}(h)$ is a~product of probabilities of the~actions taken by player $i$ between the root and $h$.
The reach probabilities $\pi_i^\sigma(h|g)$ and $\pi^\sigma(h|g)$ conditional on being in some $g\sqsubset h$ are defined analogously, except that the products are only taken over the~actions on the path between $g$ and $h$.
Finally, $\pi^{\sigma}_{-i}(\cdot)$ is defined like $\pi^\sigma(\cdot)$, except that in the product $\pi^\sigma_1(\cdot)\pi^\sigma_2(\cdot)\pi^\sigma_c(\cdot)$ the term $\pi^\sigma_i(\cdot)$ is replaced by 1.

The~\emph{expected utility} for player $i$ of a strategy profile $\sigma$ is $u_i(\sigma) = \sum_{z\in \mc Z} \pi^\sigma(z)u_i(z)$.
The~profile $\sigma$ is an \emph{$\epsilon$-Nash equilibrium} ($\epsilon$-\acs{NE}) if
\begin{align*}
\left( \forall i\in\{1,2\} \right) \ : \ u_i(\sigma) \geq \max_{\sigma_i' \in \Sigma_i} u_i(\sigma_i',\sigma_\opponent) -\epsilon .
\end{align*}
A Nash equilibrium (\ac{NE}) is an $\epsilon$-\ac{NE} with $\epsilon=0$. It is a standard result that in two-player zero-sum games, all $\sigma^*\in \textrm{\ac{NE}}$  have the same $u_i(\sigma^*)$~\cite{osborne1994course}. The \emph{exploitability} $\expl(\sigma)$ of $\sigma \in \Sigma$ is the average of exploitabilities $\expl_i(\sigma)$, $i \in \{1,2\}$, where
\begin{equation*}
    \expl_i(\sigma) := u_i(\sigma^*) - \min_{\sigma'_\opponent \in \Sigma_\opponent} u_i(\sigma_i,\sigma'_\opponent).
\end{equation*}
The expected utility conditional on reaching $h\in\mc H$ is
\begin{equation*}
    u^\sigma_i(h) = \sum_{h\sqsubset z\in \mc Z}\pi^{\sigma}(z|h)u_i(z).
\end{equation*}
An ingenious variant of this concept is the~\emph{counterfactual value} (\acs{CFV}) of a history, defined as $v_i^\sigma(h) := \pi^\sigma_{-i}(h) u_i^\sigma(h)$, and the coun\-terfactual value of taking an~action $a$ at $h$, defined as  $v_i^\sigma(h,a)  := \pi^\sigma_{-i}(h) u_i^\sigma(ha) $.
We set $v_i^\sigma(I) := \sum_{h\in I} v_i^\sigma(h)$ for $I\in \mc I_i$ and define $v_i^\sigma(I,a)$ analogously.
A strategy $\sigma^\star_2\in \Sigma_2$ is a \emph{counterfactual best response} $\textrm{CBR}(\sigma_1)$ to $\sigma_1 \in \Sigma_1$ if $v^{(\sigma_1,\sigma^\star_2)}_2(I) = \max_{a\in\mc A(I)} v^{(\sigma_1,\sigma^\star_2)}_2(I,a)$ holds for each $I \in \mc I_2$~\cite{CFR-D}. 

%

\subsection{Monte Carlo \acs{CFR}}

For a strategy $\sigma\in \Sigma$, $I\in \mc I_i$ and $a\in \mc A(I)$, we set the counterfactual regret for not playing $a$ in $I$ under strategy $\sigma$ to
\begin{equation}
 r_i^\sigma(I,a) := v_i^{\sigma} (I,a) - v_i^\sigma (I).
\end{equation}
The Counterfactual Regret minimization (CFR) algorithm~\cite{CFR} generates a consecutive sequence of strategies $\sigma^0, \sigma^1,\,\dots,\,\sigma^T$ in such a way that the \emph{immediate counterfactual regret}
\[ \bar R^t_{i,\textnormal{imm}} (I) := \!
\max_{a\in \mc A(I)} \bar R^t_{i,\textnormal{imm}} (I,a) := \!
\max_{a\in \mc A(I)} \frac{1}{t} \sum_{t'=1}^t r_i^{\sigma^{t'}}(I,a) \]
is minimized for each $I\in \mc I_i$, $i \in \{1,2\}$.
It does this by using the Regret Matching update rule~\cite{hart2000simple,blackwell}:
\begin{equation}\label{eq:rm_update}
\sigma^{t+1}(I, a) := \frac{\max \{ \bar R^{t}_{i,\textnormal{imm}} (I,a) ,0\} }{ \sum_{a' \in \mc A(I)} \max \{ \bar R^{t}_{i,\textnormal{imm}} (I,a'), 0 \} }.
\end{equation}
Since the overall regret is bounded by the sum of immediate counterfactual regrets~\cite[Theorem 3]{CFR}, this causes the average strategy $\bar{\sigma}^T$ (defined by \eqref{eq:avg_strat1}) to converge to a \ac{NE}~\cite[Theorem 1]{MCCFR}:
\begin{align}\label{eq:avg_strat1}
        \bar{\sigma}^T(I, a) := \frac{
                \sum_{t=1}^T \pi^{\sigma^t}_i(I) \sigma^t(I,a)
        }{
                \sum_{t=1}^T \pi^{\sigma^t}_i(I)
        } && \textnormal{(where $I\in\mc I_i$)} .
\end{align}
In other words, by accumulating immediate cf. regrets at each information set from the strategies $\sigma^0, \dots, \sigma^t$, we can produce new strategy $\sigma^{t+1}$. However only the average strategy is guaranteed to converge to NE with $\mc O(1/\sqrt{T})$ -- the individual regret matching strategies can oscillate.
The initial strategy $\sigma^0$ is uniform, but in fact any strategy will work. If the sum in the denominator of update rule~(\ref{eq:rm_update}) is zero, $\sigma^{t+1}(I,a)$ is set to be also uniform.

The disadvantage of CFR is the costly need to traverse the whole game tree during each iteration.
Monte Carlo CFR~\cite{MCCFR} works similarly, but only samples a small portion of the game tree each iteration. It calculates sampled variants of CFR's variables, each of which is an unbiased estimate of the original~\cite[Lemma~1]{MCCFR}.
We use a particular variant of \ac{MCCFR} called Outcome Sampling (\acs{OS})~\cite{MCCFR}.
\ac{OS} only samples a single terminal history $z$ at each iteration, using the sampling strategy $\sigma^{t,\epsilon} := (1-\epsilon)\sigma^t + \epsilon\cdot \textrm{rnd}$, where $\epsilon \in (0,1]$ controls the exploration and $\textrm{rnd}(I, a) := \frac{1}{|\mc A(I)|}$.

This $z$ is then traversed forward (to compute each player's probability $\pi_i^{\sigma^t}(h)$ of playing to reach each prefix of $z$) and backward (to compute each player's probability $\pi_i^{\sigma^t}(z|h)$ of playing the remaining actions of the history).
During the backward traversal, the sampled counterfactual regrets at each visited $I\in \mc I$ are computed according to  \eqref{eq:sampled_regret} and added to $\tilde R^T_{i,\textnormal{imm}}(I)$:
\begin{align}\label{eq:sampled_regret}
\tilde{r}^{\sigma^t}_i(I,a) :=
\begin{cases}
w_I\cdot  ( \pi^{\sigma^t}(z|ha) - \pi^{\sigma^t}(z|h) ) & \textnormal{if } h a \sqsubset z \\
w_I \cdot (0 - \pi^{\sigma^t}(z|h) ) & \textnormal{otherwise}
\end{cases},
\end{align}
where $h$ denotes the prefix of $z$ which is in $I$ and $w_I$ stands for
$\frac{1}{\pi^{\sigma^{t,\epsilon}}(z)} \pi_{-i}^{\sigma^t} (z|h) u_i(z)$~\cite{lanctot_thesis}.

\section{Domain-Independent Formulation of Continual Resolving}
The only domain for which continual resolving has been previously defined and implemented is poker. Poker has several special properties:
a) all information sets have a fixed number of histories of the same length,
b) public states have the same size and
c) only a single player is active in any public state.

There are several complications that occur in more general EFGs:
(1) We might be asked to take several turns within a single public state, for example in phantom games.
(2) When we are not the acting player, we might be unsure whether it is the opponent's or chance's turn.
(3) Finally, both players might be acting within the same public state, for example a secret chance roll determines whether we get to act or not.

In this section, we present an~abstract formulation of continual resolving robust enough to handle the complexities of general EFGs.
However, we first need to define the~concepts like public tree and resolving gadget more carefully.

\subsection{Subgames and the Public Tree}

To speak about the information available to player $i$ in histories where he doesn't act, we use \emph{augmented information sets}.
For player $i\in \{1,2\}$ and history $h\in \mc H\setminus \mc Z$, the $i$'s \emph{observation history} $\vec O_i(h)$ in $h$ is the sequence $(I_1,a_1,I_2,a_2,~\dots)$ of the information sets visited and actions taken by $i$ on the path to $h$ (incl. $I\ni h$ if $h\in \mc H_i$). Two histories $g,h\in \mc H\setminus \mc Z$ belong to the same \emph{augmented information set} $I\in \mc I_i^{\textnormal{aug}}$ if $\vec O_i(g) = \vec O_i(h)$.
This is equivalent to the definition from~\citep{CFR-D}, except that our definition makes it clear that $\mc I_i^{\textnormal{aug}}$ is also defined on $\mc H_i$ (and coincides there with $\mc I_i$ because of perfect recall).

\begin{remark}[Alternatives to $\mc I^{\textnormal{aug}}$]
$\mc I^{\textnormal{aug}}$ isn't the only viable way of generalizing information sets. One could alternatively consider some further-unrefineable perfect-recall partition $\mc I^*_i$ of $\mc H$ which coincides with $\mc I_i$ on $\mc H_i$, and many other variants between the two extremes. We focus only on $\mc I^{\textnormal{aug}}$, since an in-depth discussion of the general topic would be outside of the scope of this paper.
\end{remark}


We use $\sim$ to denote histories indistinguishable by some player:
\[ g\! \sim \! h \iff \vec O_1(g)=\vec O_1(h) \lor \vec O_2(g)=\vec O_2(h) .\]
By $\approx$ we denote the~transitive closure of $\sim$. Formally, $g \approx h$ iff
\[
\left( \exists n \right) \left( \exists h_1, \dots, h_n \right) :
g\!\sim\! h_1, \ h_1 \!\sim\! h_2, \ \dots,\ h_{n-1} \!\sim\! h_n, \ h_n \!\sim\! h .
\]

\noindent If two states do \emph{not} satisfy $g \approx h$, then it is common knowledge that both players can tell them apart.

\begin{definition}[Public state] \label{def:publ_state}
\emph{Public partition} is any partition $\mc S$ of $\mc H\setminus \mc Z$ whose elements are closed under $\sim$ and form a tree. 
An~element $S$ of such $\mc S$ is called a \emph{public state}.
The \emph{common knowledge partition} $\mc S_\textrm{ck}$ is the one consisting of the equivalence classes of $\approx$.
\end{definition}

\noindent Our definition of $\mc S$ is a~reformulation of the definition of~\cite{accelerated_BR} in terms of augmented information sets (which aren't used in~\cite{accelerated_BR}). The~definition of $\mc S_\textnormal{ck}$ is novel.
We endow any $\mc S$ with the tree structure inherited from $\mc H$. Clearly, $\mc S_\textrm{ck}$ is the finest public partition.
The~concept of a~public state is helpful for introducing imperfect-information subgames (which aren't defined in~\cite{accelerated_BR}).
\begin{definition}[Subgame]\label{def:subgame}
    A \emph{subgame} rooted at a~public state $S$ is the set $G(S) := \{ h\in \mc H | \ \exists g \in S: g\sqsubset h \}$.
\end{definition}

For comparison,~\cite{CFR-D} defines a subgame as ``a forest of trees, closed under both the descendant relation and membership within $\aug_i$ for any player''.
For any $h\in S \in \mc S_\textrm{ck}$, the subgame rooted at $S$ is the smallest~\citep{CFR-D}-subgame containing $h$.
As a result,~\cite{CFR-D}-subgames are ``forests of subgames rooted at common-knowledge public states''.

We can see that finer public partitions lead to smaller subgames, which are easier to solve. In this sense, the common-knowledge partition is the ``best one''. However, finding $\mc S_\textrm{ck}$ is sometimes non-trivial, which makes the~definition of general public states from~\cite{accelerated_BR} important.
The drawback of this definition is its ambiguity --- indeed, it allows for extremes such as grouping the whole $\mc H$ into a single public state, without giving a practical recipe for arriving at the ``intuitively correct'' public partition.

\subsection{Aggregation and the Upper Frontier}

Often, it is useful to aggregate reach probabilities and counterfactual values over (augmented) information sets or public states. In general EFGs, an augmented information set $I\in \mc I_i^\textnormal{aug}$ can be ``thick'', i.e. it can contain both some $ha\in \mc H$ and it's parent $h$.
This necessarily happens when we are unsure how many actions were taken by other players between our two successive actions.
For such $I$, we only aggregate over the ``upper frontier'' $\hat I := \{ h \in I | \, \nexists g \in I: g \sqsubset h \, \& \, g\neq h \}$ of $I$~\cite{halpern2016upperFrontier,halpern1997upperFrontier}: 
We overload $\pi^\sigma(\cdot)$ as $\pi^\sigma(I) := \sum_{h\in\hat I} \pi^\sigma(h)$ and $v_i^\sigma(\cdot)$ as $v_i^\sigma(I) := \sum_{h\in\hat I} v_i^\sigma(h)$. We define $\hat S$ for $S\in \mc S$, $\pi_i^\sigma(I)$, $\pi_{-i}^\sigma(I)$ and $v^\sigma_i(I,a)$ analogously.
By $\hat S(i) := \{  I \in \mc I_i^\textnormal{aug}\, | \ \hat I \subseteq \hat S \}$ we denote the topmost (augmented) information sets of player $i$ in $S$.

To the best of our knowledge, the issue of ``thick'' information sets has only been discussed in the context of non-augmented information sets in games with imperfect recall~\cite{halpern2016upperFrontier}.
One scenario where thick augmented information sets cause problems is the resolving gadget game, which we discuss next.

\subsection{Resolving Gadget Game}\label{sec:gadget}
We describe a generalization of the resolving gadget game from~\cite{CFR-D} (cf.~\cite{MoravcikGadget,brown2017safe}) for resolving Player 1's strategy (see Figure~\ref{fig:gadget}).

Let $S\in \mc S$ be a public state to resolve from, $\sigma\in \Sigma$, and let $\tilde v(I) \in \R$ for $I\in \hat S(i)$ be the required counterfactual values.
First, the upper frontier of $S$ is duplicated as $\{ \tilde h | \, h\in \hat S \} =: \tilde S$.
Player 2 is the acting player in $\tilde S$, and from his point of view, nodes $\tilde h$ are partitioned according to $\{\tilde I := \{ \tilde h| \, h\in \hat I\}\  |\ I\in \hat S(2)\}$.
In $\tilde h \in \tilde I$ corresponding to $h\in I$, he can choose between ``following'' (F) into $h$ and ``terminating'' (T), which ends the game with utility $\tilde u_2(\tilde h  T) := \tilde v(I) \pi_{-2}^\sigma(S) / \pi^\sigma_{-2}(I)$.
From any $h\in \hat S$ onward, the game is identical to $G(S)$, except that the utilities are multiplied by a constant: $\tilde u_i(z) := u_i(z) \pi_{-2}^\sigma(S)$.
To turn this into a well-defined game, a root chance node is added and connected to each $h\in \hat S$, with transition probabilities $\pi_{-2}^\sigma(h) / \pi_{-2}^\sigma(S)$.

\begin{figure}[t]
        \includegraphics[width=0.75\linewidth]{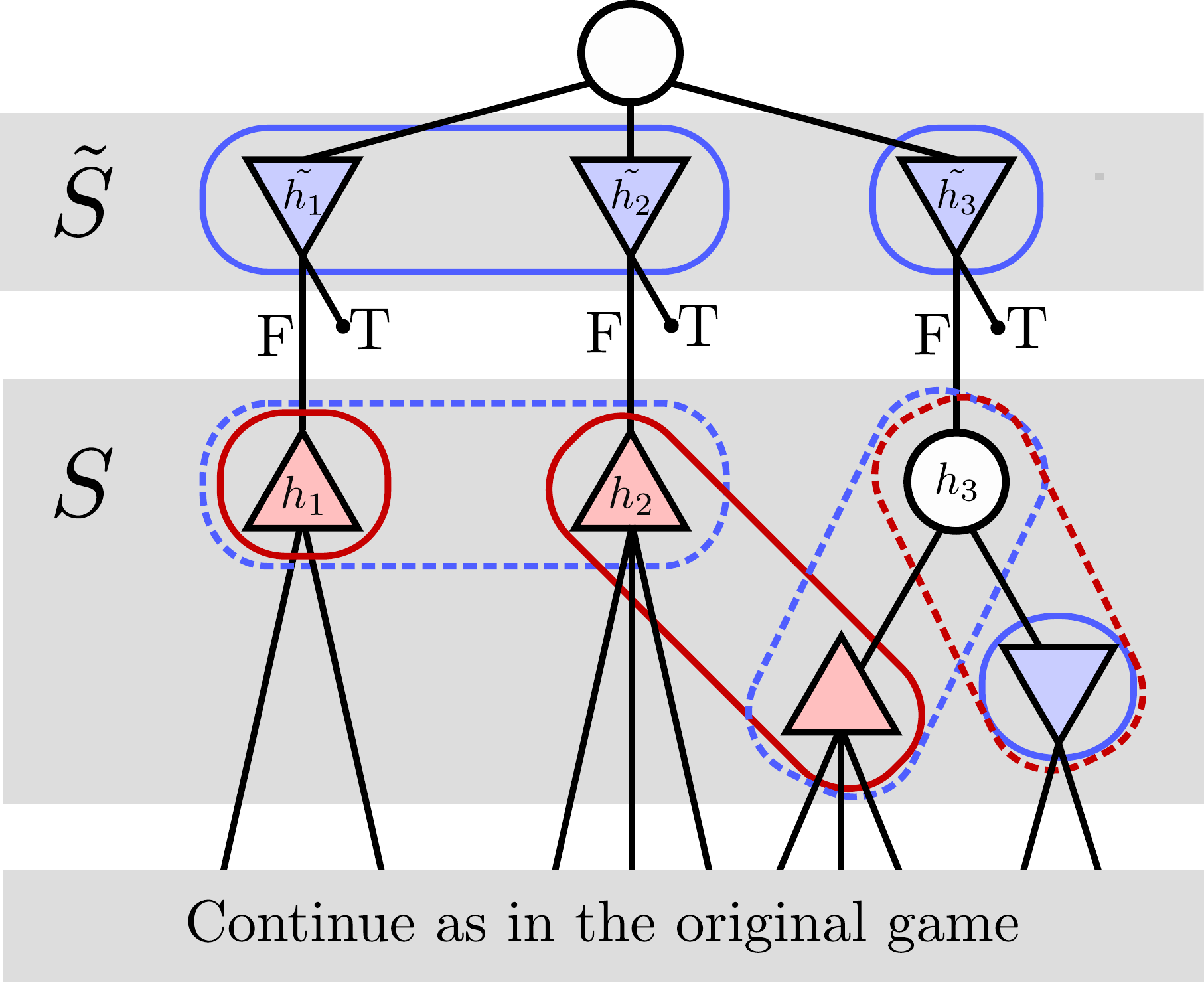}
        \caption{Resolving game $\widetilde G\left( S, \sigma, \tilde v \right)$ constructed for player $\triangle$ in a~public state $S$. Player's (augmented) information sets are drawn with solid (resp. dashed) lines of the respective color. The chance node $\bigcirc$ chooses one of $\triangledown$'s histories $\tilde{h_1},\tilde{h_2},\tilde{h_3}$, which correspond to the ``upper frontier'' of $S$.}
        \label{fig:gadget}
\end{figure}

This game is called the \emph{resolving gadget game} $\widetilde G\left( S, \sigma, \tilde v \right)$, or simply $\widetilde G\left(S\right)$ when there is no risk of confusion, and the variables related to it are denoted by tilde.
If $\tilde \rho \in \widetilde \Sigma$ is a ``resolved'' strategy in $\widetilde G(S)$, we denote the new combined strategy in $G$ as $\sigma^{\textnormal{new}} := \sigma|_{G(S)\leftarrow \tilde \rho}$, i.e. play according to strategy $\tilde \rho$ in the subgame $G(S)$ and according to $\sigma$ everywhere else.

The difference between $\widetilde G\left( S, \sigma, \tilde v \right)$ and the original version of~\cite{CFR-D} is that our modification only duplicates the upper frontier $\hat S$ and uses normalization constant $\sum_{\hat S} \pi_{-2}^\sigma (h)$ (rather than $\sum_{S} \pi_{-2}^\sigma (h)$) and estimates $\tilde v(I)=\sum_{\hat I} \tilde v(h)$ (rather than $\sum_{I} \tilde v(h)$).
This enables $\widetilde G\left( S, \sigma, \tilde v \right)$ to handle domains with thick information sets and public states.
While tedious, it is straightforward to check that $\widetilde G\left( S, \sigma, \tilde v \right)$  has all the properties proved in~\cite{CFR-D,DeepStack,Neil_thesis}.
Without our modification, the resolving games would either sometimes be ill-defined, or wouldn't have the~desired properties.

The following properties are helpful to get an intuitive understanding of gadget games. Their more general versions and proofs (resp. references for proofs) are listed in the appendix.

\begin{lemma}[Gadget game preserves opponent's values]
For each $I \in \mc I_2^\textnormal{aug}$ with $I\subset G(S)$, we have
$v_2^{\sigma^\textnormal{new}}(I) = \tilde v_2^{\tilde \rho}(I)$.
\end{lemma}

\noindent Note that the conclusion does \emph{not} hold for counterfactual values of the (resolving) player 1! (This can be easily verified on a simple specific example such as Matching Pennies.)

\begin{lemma}[Optimal resolving]
If $\sigma$ and $\tilde \rho$ are both Nash equilibria and $\tilde v(I) = v_2^\sigma(I)$ for each $I\in \hat S(2)$, then $\sigma^{\textnormal{new}}_1$ is not exploitable.
\end{lemma}

\subsection{Continual Resolving}

Domain-independent continual resolving closely follows the structure of continual resolving for poker~\cite{DeepStack}, but uses a~generalized resolving gadget and handles situations which do not arise in poker, such as multiple moves in one public state. We explain it from the perspective of Player 1. The abstract CR keeps track of strategy $\sigma_1$ it has computed in previous moves. Whenever it gets to a public state $S$, where $\sigma_1$ has not been computed, it resolves the subgame $G(S)$.
As a by-product of this resolving, it estimates opponent's counterfactual values $v_2^{\sigma_1,\textrm{CBR}(\sigma_1)}$ for all public states that might come next, allowing it to keep resolving as the game progresses.

CR repetitively calls a \texttt{Play} function which takes the current information set $I\in \mc I_1$ as the input and returns an action $a\in \mc A(I)$ for Player 1 to play. It maintains the following variables:
\begin{itemize}
	\item $S\in \mc S$ \dots the current public state,
	\item $\textrm{KPS} \subset \mc S$ \dots the public states where strategy is known,
	\item $\sigma_1$ \dots a strategy defined for every $I\in \mc I_1$ in \textrm{KPS},
	\item $\textrm{NPS} \subset \mc S$ \dots the public states where \texttt{CR} may resolve next,
	\item $D(S')$ for $S'\in \textrm{NPS}$ \dots data allowing resolving at $S'$, such as the estimates of opponent's counterfactual values.
\end{itemize}

\begin{algorithm}[t]
\LinesNumbered
\SetKwData{knownStrategy}{KPS}
\SetKwData{knownValues}{NPS}
\SetKwData{state}{S}
\SetKwData{data}{D}
\SetKwData{action}{a}

\SetKwFunction{BuildGadget}{BuildResolvingGame}
\SetKwFunction{GetNextStrategy}{ExtendKPS}
\SetKwFunction{Resolve}{Resolve}

\SetKwFunction{assert}{assert}
\SetKwFunction{return}{return}

\SetKwInOut{Input}{Input}
\SetKwInOut{Output}{Output}

\Input{Information set $I\in \mc I_1$}
\Output{An action $\action\in \mc A(I)$}
\BlankLine
\state $\leftarrow$ the public state which contains $I$\;
\If{\state $\notin$ \knownStrategy}{
	$\widetilde G(S) \leftarrow$ \BuildGadget{\state,\data{\state}}\;
	\knownStrategy $\leftarrow$ \knownStrategy $\cup$ \state\;
	\knownValues $\leftarrow$ all $S'\in \mc S$ where CR acts for the first time after leaving $\knownStrategy$\;\label{line:NPS}
	$\tilde \rho, \tilde \data \leftarrow$ \Resolve{$\widetilde G(S)$,\knownValues}\;
	$\sigma_1 |_{S'} \leftarrow \tilde \rho |_{S'}$\;
	\data $\leftarrow$ calculate data for \knownValues based on \data, $\sigma_1$ and $\tilde \data$\;
}
\textbf{return} $\action\sim \sigma_1(I)$
\caption{Function \texttt{Play} of Continual Resolving}\label{alg:CR}
\end{algorithm}

The pseudo-code for CR is described in Algorithm~\ref{alg:CR}.
If the current public state belongs to \textrm{KPS}, then the strategy $\sigma_1(I)$ is defined, and we sample action $a$ from it.
Otherwise, we should have the data necessary to build some resolving game $\widetilde G(S)$ (line 3).
We then determine the public states $\textrm{NPS}$ where we might need to resolve next (line 5).
We solve $\widetilde G(S)$ via some resolving method which also computes the data necessary to resolve from any $S'\in \textrm{NPS}$ (line 6).
Finally, we save the resolved strategy in $S$ and update the data needed for resolving (line 7-9).
To initialize the variables before the first resolving, we
set \textrm{KPS} and $\sigma_1$ to $\emptyset$, find appropriate $\textrm{NPS}$, and start solving the game from the root using the same solver as \texttt{Play}, i.e. $\_\, , D \leftarrow \texttt{Resolve}(G, \textrm{NPS})$.

We now consider CR variants that use the gadget game from Section~\ref{sec:gadget} and data of the form $D=(r_1,\tilde v)$, where $r_1(S') = ( \pi_1^{\sigma_1}(h) )_{S'}$ is CR's range and $\tilde v(S') = (\tilde v(J))_J$ estimates opponent's counterfactual value at each $J\in S'(2)$.
We shall use the following notation: $S_n$ is the $n$-th public state from which CR resolves; $\tilde \rho_n$ is the corresponding strategy in $\widetilde G(S_n)$; $\sigma_1^n$ is CR's strategy after $n$-th resolving, defined on $\textrm{KPS}_n$; the optimal extension of $\sigma_1^n$ is
$$\sigma_1^{*n} := {\arg \! \min}_{\nu_1 \in \Sigma_1} \ \textnormal{expl}_1 \left( \sigma_1^n|_{\textrm{KPS}_n} \cup \nu_1|_{\mc S \setminus \textrm{KPS}_n} \right) .$$
Lemmata 24 and 25 of~\cite{DeepStack} (summarized into Lemma~\ref{lem:resolving_lemma} in our Appendix~\ref{sec:proofs}) give the following generalization of~\cite[Theorem\,S1]{DeepStack}:

\begin{theorem}[Continual resolving bound]\label{thm:cr}
Suppose that CR uses $D=(r_1,\tilde v)$ and $\widetilde G(S,\sigma_1,\tilde v)$.
Then the exploitability of its strategy is bounded by
$\textnormal{expl}_1 (\sigma_1) \leq \epsilon_0^{\tilde v} + \epsilon_1^R + \epsilon_1^{\tilde v} + \dots + \epsilon_{N-1}^{\tilde v} + \epsilon_N^R$,
where $N$ is the number of resolving steps and
$\epsilon_n^R := \widetilde{\textnormal{expl}}_1(\tilde \rho_n)$,
$\epsilon_n^{\tilde v} := \sum_{J\in \hat S_{n+1}(2)} \left| \tilde v(J) - v_2^{\sigma_1^{*n},CBR}(J) \right|$
are the exploitability (in $\widetilde G(S_n)$) and value estimation error made by the $n$-th resolver (resp. initialization for $n=0$).
\end{theorem}

The DeepStack algorithm from~\cite{DeepStack} is a~poker-specific instance of the~CR variant described in the above paragraph. Its resolver is a modification of CFR with neural network heuristics and sparse look-ahead trees. We make CR domain-independent and allowing for different resolving games ($\texttt{BuildResolvingGame}$), algorithms ($\texttt{Resolve}$), and schemes (by changing line \ref{line:NPS}).

\section{Monte Carlo Continual Resolving}\label{sec:mccr}


Monte Carlo Continual Resolving is a specific instance of CR which uses Outcome Sampling \ac{MCCFR} for game (re)solving. Its data are of the form $D=(r_1,\tilde v)$ described above and it resolves using the gadget game from Section~\ref{sec:gadget}. We first present an abstract version of the algorithms that we formally analyze, and then add improvements that make it practical.
To simplify the theoretical analysis, we assume \ac{MCCFR} computes the exact counterfactual value of resulting average strategy $\bar \sigma^T$ for further resolving. (We later discuss more realistic alternatives.)
The following theorem shows that \ac{MCCR}'s exploitability converges to 0 at the rate of $O(T^{-1/2})$.

\begin{restatable}[MCCR bound]{theorem}{MCCRbound}\label{thm:mccr}
With probability at least $(1-p)^{N+1}$, the exploitability of strategy $\sigma$ computed by \ac{MCCR} satisfies
\begin{align*}
    \expl_i(\sigma) \leq
    \left( \sqrt{2} / \sqrt{p} + 1 \right)|\mc I_i|
    \frac{\Delta_{u,i}\sqrt{A_i}}{\delta}
    \left( \frac{2}{\sqrt{T_0}} + \frac{2N-1}{\sqrt{T_R}} \right),
\end{align*}
where $T_0$ and $T_R$ are the numbers of \ac{MCCR}'s iterations in pre-play and each resolving, $N$ is the required number of resolvings, $\delta = \min_{z, t} q_t(z)$ where $q_t(z)$ is the probability of sampling $z\in \mc Z$ at iteration $t$, $\Delta_{u,i} = \max_{z,z'} | u_i(z) - u_i(z')|$ and $A_i = \max_{I\in \mc I_i} |\mc A(I)|$.
\end{restatable}

The proof is presented in the appendix. Essentially, it inductively combines the \ac{OS} bound (Lemma~\ref{lem:MCCFR_expl}) with the guarantees available for resolving games in order to compute the overall exploitability bound.\footnote{Note that Theorem~\ref{thm:mccr} isn't a straightforward corollary of Theorem~\ref{thm:cr}, since calculating the numbers $\epsilon_n^{\tilde v}$ does require non-trivial work. In particular, $\bar \sigma^T$ from the $n$-th resolving isn't the same as $\sigma_1^{n*}, CBR(\sigma_1^{n*})$ and the simplifying assumption about $\tilde v$ is \emph{not} equivalent to assuming that $\epsilon_n^{\tilde v} = 0$.}
For specific domains, a much tighter bound can be obtained by going through our proof in more detail and noting that the size of subgames decreases exponentially as the game progresses (whereas the proof assumes that it remains constant). In effect, this would replace the $N$ in the bound above by a small constant.

\subsection{Practical Modifications}\label{sec:hacks}

Above, we describe an abstract version of \ac{MCCR} optimized for clarity and ease of subsequent theoretical analysis.
We now describe the version of \ac{MCCR} that we implemented in practice. The code used for the experiments is available online at \url{https://github.com/aicenter/gtlibrary-java/tree/mccr}.

\subsubsection{Incremental Tree-Building} A massive reduction in the memory requirements can be achieved by building the game tree incrementally, similarly to Monte Carlo Tree Search (\acs{MCTS})~\cite{browne2012survey}. We start with a tree that only contains the root.
When an information set is reached that is not in memory, it is added to it and a playout policy (e.g., uniformly random) is used for the remainder of the sample.
In playout, information sets are not added to memory. Only the regrets in information sets stored in the memory are updated.

\subsubsection{Counterfactual Value Estimation}\label{sec:cfv-estimation}

Since the computation of the~exact counterfactual values of the average strategy needed by $\widetilde G(S,\sigma,\cdot)$ requires the traversal of the whole game tree, we have to work with their estimates instead.
To this end, our \ac{MCCFR} additionally computes the opponent's \emph{sampled counterfactual values}
\[ \tilde v_2^{\sigma^t}(I) :=
\frac{1}{\pi^{\sigma^{t,\epsilon}}(z)}  \pi^{\sigma^t}_{-2}(h)  \pi^{\sigma^t} \! (z|h)  u_2(z)
.\]
It is not possible to compute the exact counterfactual value of the average strategy just from the values of the current strategies. Once the $T$ iterations are complete, the standard way of estimating the counterfactual values of $\bar \sigma^T$ is using \emph{arithmetic averages}
\begin{equation}\label{eq:ord-avg-sampl-vals}
\tilde v(I) := \frac{1}{T}\sum \tilde v_2^{\sigma^t}(I).
\end{equation}
However, we have observed better results with \emph{weighted averages}
\begin{equation}\label{eq:weigh-avg-sampl-vals}
\tilde v(h) := \sum_t \tilde \pi^{\sigma^t}\!\!(h) \, v_2^{\sigma^t}\!(h) \ / \ \sum_t \tilde \pi^{\sigma^t}\!\!(h).
\end{equation}
The stability and accuracy of these estimates is experimentally evaluated in Section~\ref{sec:experiments} and further analyzed in Appendix~\ref{sec:cf_values}. We also propose an unbiased estimate of the exact values computed from the already executed samples, but its variance makes it impractical.

\subsubsection{Root Distribution of Gadgets}
As in~\cite{DeepStack}, we use the information about opponent's strategy from previous resolving when constructing the gadget game. Rather than being proportional to $\pi_{-2}(h)$, the root probabilities are proportional to $\pi_{-2}(h)(\pi_2(h)+\epsilon)$. This modification is sound as long as $\epsilon>0$.

\subsubsection{Custom Sampling Scheme}\label{sec:custom_sampling}
To improve the efficiency of resolving by \ac{MCCFR}, we use a custom sampling scheme which differs from \ac{OS} in two aspects.
First, we modify the above sampling scheme such that with probability 90\% we sample a history that belongs to the current information set $I$. This allows us to focus on the most relevant part of the game.
Second, whenever $\tilde h \in \tilde S$ is visited by \ac{MCCFR}, we sample both actions (T and F). This increases the transparency of the algorithm, since all iterations now ``do a similar amount of work'' (rather than some terminating immediately).
These modifications are theoretically sound, since the resulting sampling scheme still satisfies the assumptions of the general \ac{MCCFR} bound from~\cite{lanctot_thesis}.

\subsubsection{Keeping the Data between Successive Resolvings}
Both in pre-play and subsequent resolvings, \ac{MCCFR} operates on successively smaller and smaller subsets of the game tree. In particular, we don't need to start each resolving from scratch, but we can re-use the previous computations.
To do this, we initialize each resolving \ac{MCCFR} with the \ac{MCCFR} variables (regrets, average strategy and the corresponding value estimates) from the previous resolving (resp. pre-play). In practice this is accomplished by simply not resetting the data from the previous \ac{MCCFR}.
While not being backed up by theory, this approach worked better in most practical scenarios, and we believe it can be made correct with the use of warm-starting~\cite{warm_start} of the resolving gadget.


\section{Experimental evaluation} \label{sec:experiments}

After brief introduction of competing methods and explaining the used methodology, we focus on evaluating the alternative methods to estimate the counterfactual values required for resolving during \ac{MCCFR}.
Next, we evaluate how quickly and reliably these values can be estimated in different domains, since these values are crucial for good performance of \ac{MCCR}. Finally, we compare exploitability and head-to-head performance to competing methods.

\subsection{Competing Methods}
\noindent \textbf{Information-Set Monte Carlo Tree Search.\ }
\acs{IS-MCTS}~\cite{ISMCTS} runs \ac{MCTS} samples as in a perfect information game, but computes statistics for the whole information set and not individual states.
When initiated from a non-empty match history, it starts samples uniformly
from the states in the current information set.
We use two selection functions: Upper Confidence bound applied to Trees (\acs{UCT})~\cite{kocsis2006bandit} and Regret Matching (\acs{RM})~\cite{hart2001reinforcement}. We use the same settings as~in~\cite{OOS}: UCT constant 2x the maximal game outcome, and RM with exploration 0.2.
In the following, we refer to \ac{IS-MCTS} with the corresponding selection function by only \texttt{UCT} or \texttt{RM}.

\noindent \textbf{Online Outcome Sampling.\ }
\ac{OOS}~\cite{OOS} is an online search variant of \ac{MCCFR}.
\ac{MCCFR} samples from the root of the tree and needs to pre-build the whole game tree.
\ac{OOS} has two primary distinctions from \ac{MCCFR}: it builds its search tree incrementally and
it can bias samples with some probability to any specific parts of the game tree.
This is used to target the information sets (\acs{OOS-IST}) or the public states (\acs{OOS-PST})
where the players act during a match.

We do not run OOS-PST on domain of IIGS, due to non-trivial biasing of sampling towards current public state.


We further compare to \ac{MCCFR} with incremental tree building and the random player denoted RND.
\subsection{Computing Exploitability}\label{sec:perf-eval}

Since the online game playing algorithms do not compute the strategy for the whole game, evaluating exploitability of the strategies they play is more complicated. 
One approach, called brute-force in~\cite{OOS}, suggest ensuring that the online algorithm is executed in each information set in the game and combining the computed strategies. If the thinking time of the algorithm per move is $t$, it requires $O(t\dot|\mc I|)$ time to compute one combined strategy and multiple runs are required to achieve statistical significance for randomized algorithms. While this is prohibitively expensive even for the smaller games used in our evaluation, computing the strategy for each public state, instead of each information set is already feasible. We use this approach, however, it means we have to disable the additional targeting of the current information set in the resolving gadget proposed in Section~\ref{sec:custom_sampling}.

There are two options how to deal with the variance in the combined strategies in different runs of the algorithm in order to compute the exploitability of the real strategy realized by the algorithm. The pessimistic option is to compute the exploitability of each combined strategy and average the results. This assumes the opponent knows the random numbers used by the algorithm for sampling in each resolving. A more realistic option is to average the combined strategies from different runs into an \emph{expected strategy} $\dbbar{\sigma}$ and compute its exploitability. We use the latter.

\subsection{Domains}
For direct comparison with prior work, we use same domains as~in~\cite{OOS} with parametrizations noted in parentheses: Imperfect Information Goofspiel \texttt{IIGS(N)}, Liar's Dice \texttt{LD(D1,D2,F)} and Generic Poker \texttt{GP(T,C,R,B)}. We add Phantom Tic-Tac-Toe \texttt{PTTT} to also have a domain with thick public states, and use Biased Rock Paper Scissors \texttt{B-RPS}~for small experiments. The detailed rules are in Appendix~\ref{sec:rules} with the sizes of the domains in Table~\ref{tab:sizes}.
We use small and large variants of the domains based on their parametrization. Note that large variants are about $10^{4}$ up to $10^{15}$ times larger than the smaller ones.

\subsection{Results}
\setcounter{figure}{1}
\begin{figure*}[t]
    \includegraphics[width=0.4\linewidth]{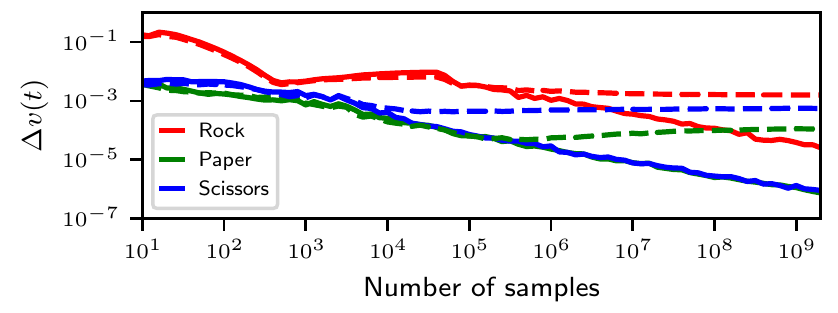}\hfill
    \includegraphics[width=0.59\linewidth]{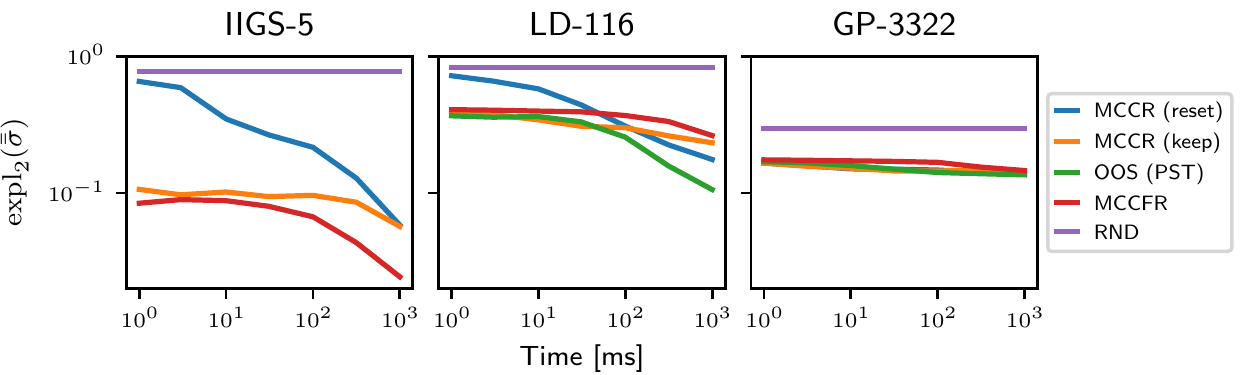}
    \caption{
    Left -- Estimation error of the arithmetic (dashed lines) and weighted averages (solid lines) of action values in \texttt{B-RPS}. Right -- Exploitability of $\dbbar{\sigma}$ as a function of the resolving time. All algorithms have pre-play of 300ms.}\label{fig:expl}
\end{figure*}

\setcounter{figure}{2}
\begin{figure}[t]
    \includegraphics[width=0.9\linewidth]{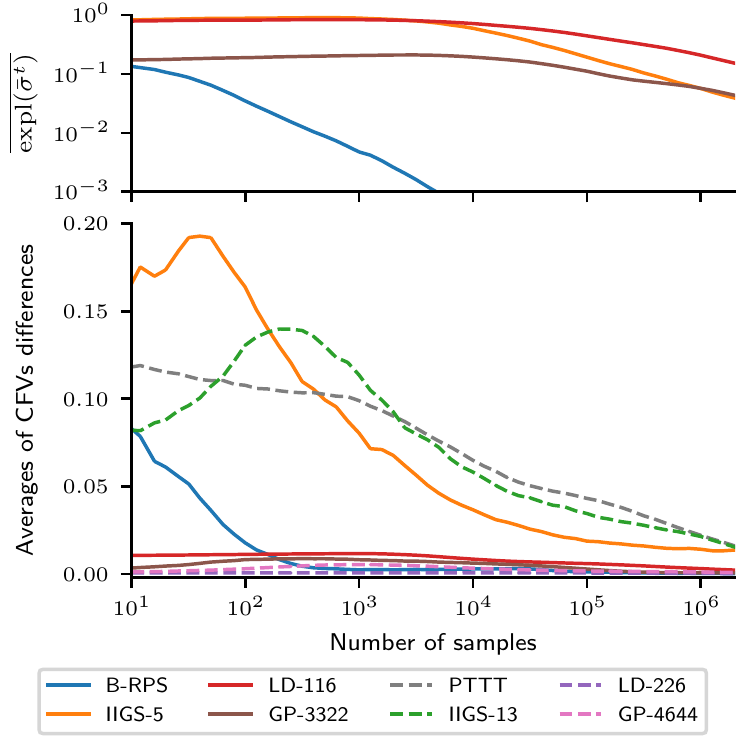}
    \caption{A comparison of exploitability (top) with ``CFV instability'' (bottom) in different domains. For $t=10^7$, the differences are 0 by definition.}\label{fig:cfv_stability}
\end{figure}

\setcounter{figure}{3}

\subsubsection{Averaging of Sampled \ac{CFV}s}\label{sec:experiment-cfv-averaging}

As mentioned in Section~\ref{sec:cfv-estimation}, computing the  exact counterfactual values of the average strategy $\bar\sigma^T$ is often prohibitive, and we replace it by using the arithmetic or weighted averages instead.
To compare the two approaches, we run \ac{MCCFR} on the B-RPS domain (which only has a single \ac{NE} $\sigma*$ to which $\bar\sigma^T$ converges) and measure the distance $\Delta v(t)$ between the estimates and the correct action-values $v_1^{\sigma*}(\textnormal{root},a)$.
%
In~Figure~\ref{fig:expl} (left), the weighted averages converge to the correct values with increasing number of samples, while the arithmetic averages eventually start diverging.
The weighted averages are more accurate (see Appendix, Figure~\ref{fig:cfv_comparison_domains} for comparison on each domain) and we will use them exclusively in following experiments.


\subsubsection{Stability of \ac{CFV}s}\label{sec:experiment-cfv-convergence}

To find a nearly optimal strategy when resolving, \ac{MCCR} first needs \ac{CFV}s of such a strategy (Lemma~\ref{lem:resolving_lemma}).
However, the \ac{MCCFR} resolver typically won't have enough time to find such $\bar \sigma^T$.
If \ac{MCCR} is to work, the \ac{CFV}s computed by \ac{MCCFR} need to be close to those of an approximate equilibrium \ac{CFV}s even though they correspond to an exploitable strategy.

We run \ac{MCCFR} in the root and focus on \ac{CFV}s in the public states where player 1 acts for the 2nd time (the gadget isn't needed for the first action, and thus neither are CFVs):
\begin{align*}\label{eq:omega}
    \Omega := \!\{ J \in \mc I_2^\textnormal{aug} \mid \exists S' \!\in \mc S, h \in S':  J \subset S' \& \textnormal{ pl. 1 played once in }h \}.
\end{align*}
Since there are multiple equilibria \ac{MCCFR} might converge to, we cannot measure the convergence exactly.
Instead, we measure the ``instability'' of \ac{CFV} estimates by saving $\tilde v_2^t(J)$ for $t\leq T$, tracking how $\Delta_t(J) := | \tilde v_2^t(J) - \tilde v_2^T(J) | $ changes over time,
and aggregating it into $\frac{1}{|\Omega|}\sum_J \Delta_t(J)$.
We repeat this experiment with 100 different \ac{MCCFR} seeds and measure the expectation of the aggregates and, for the small domains, the expected exploitability of $\bar \sigma^t$.
If the resulting ``instability'' is close to zero after $10^5$ samples (our typical time budget), we consider CFVs to be sufficiently stable.

Figure~\ref{fig:cfv_stability} confirms that in small domains (LD, GP), \ac{CFV}s stabilize long before the exploitability of the strategy gets low.
The error still decreases in larger games (GS, PTTT), but at a slow rate.

\subsubsection{Comparison of Exploitability with Other Algorithms} \label{sec:expl-other-algs}
We compare the exploitability of \ac{MCCR} to \ac{OOS-PST} and \ac{MCCFR}, and include random player for reference.
We do not include \ac{OOS-IST}, whose exploitability is comparable to that of \ac{OOS-PST}~\cite{OOS}.
For an evaluation of \ac{IS-MCTS}'s exploitability (which is very high, with the exception of poker) we refer the reader to~\cite{OOS,ISMCTS}.

Figure~\ref{fig:expl} (right) confirms that for all algorithms, the exploitability decreases with the increased time per move. \ac{MCCR} is better than \ac{MCCFR} on LD and worse on IIGS. The keep variant of \ac{MCCR} is initially less exploitable than the reset variant, but improves slower. This suggests the keep variant could be improved.

\subsubsection{Influence of the Exploration Rate}

One of \ac{MCCR}'s parameters is the exploration rate $\epsilon$ of its \ac{MCCFR} resolver. When measuring the exploitability of \ac{MCCR} we observed no noteworthy influence of $\epsilon$ (for $\epsilon = 0.2,\,0.4,\,0.6,\,0.8$,  across all of the evaluated domains).

\subsubsection{Head-to-head Performance}
For each pair of algorithms, thousands of matches have been played on each domain, alternating positions of player 1 and 2.
In the smaller (larger) domains, players have $0.3s$ (5s) of pre-play computation and $0.1s$ (1s) per move. Table~\ref{tab:matches} summarizes the~results as percentages of the maximum domain payoff.

Note that the results of the matches are not necessarily transitive, since they are not representative of the algorithms' exploitability. When computationally tractable, the previous experiment~\ref{sec:expl-other-algs} is therefore a~much better comparison of an algorithm's performance.

Both variants of MCCR significantly outperform the random opponent in all games. With the exception of PTTT, they are also at least as good as the MCCFR ``baseline''.
This is because public states in PTTT represent the number of moves made, which results in a non-branching public tree and resolving games occupy the entire level as in the original game.
\ac{MCCR} is better than \ac{OOS-PST} in LD and GP, and better than \ac{OOS-IST} in large IIGS.
MCCR is worse than IS-MCTS on all games with the exception of small LD.
However, this does not necessarily mean that MCCR's exploitability is higher than for IS-MCTS in the larger domains, MCCR only fails to find the strategy that would exploit IS-MCTS.

\begin{table*}[t!]
    \scriptsize
    \setlength\tabcolsep{1.5pt}
    \begin{subtable}[t]{.5\linewidth}%
        \centering%
        \scalebox{1.02}{\begin{tabular}{lrrrrrrl}
\toprule
\textbf{IIGS-5} &                                                              MCCR (reset) &                                                                     MCCFR &                                           OOS-PST &                                                  OOS-IST &                                                       RM &                                                      UCT &                                                     RND \\
\midrule
MCCR (keep)  &  {\color{insignwin} \nobreakspace\nobreakspace0.4 $\pm$ \nobreakspace1.2} &  {\color{insignwin} \nobreakspace\nobreakspace0.6 $\pm$ \nobreakspace1.3} & \multicolumn{1}{c}{--}  &  {\color{lose} \nobreakspace-2.6 $\pm$ \nobreakspace1.3} &  {\color{lose} \nobreakspace-2.8 $\pm$ \nobreakspace1.3} &  {\color{lose} \nobreakspace-6.5 $\pm$ \nobreakspace1.2} &  {\color{win} \nobreakspace51.2 $\pm$ \nobreakspace1.2} \\
MCCR (reset) &   &  {\color{insignlose} \nobreakspace-0.8 $\pm$ \nobreakspace1.2} & \multicolumn{1}{c}{--}  &  {\color{lose} \nobreakspace-2.5 $\pm$ \nobreakspace1.2} &  {\color{lose} \nobreakspace-5.5 $\pm$ \nobreakspace1.2} &  {\color{lose} \nobreakspace-8.1 $\pm$ \nobreakspace1.2} &  {\color{win} \nobreakspace49.1 $\pm$ \nobreakspace1.2} \\
MCCFR        &   &   & \multicolumn{1}{c}{--}  &  {\color{lose} \nobreakspace-1.3 $\pm$ \nobreakspace1.3} &  {\color{lose} \nobreakspace-2.7 $\pm$ \nobreakspace1.3} &  {\color{lose} \nobreakspace-5.6 $\pm$ \nobreakspace1.2} &  {\color{win} \nobreakspace48.5 $\pm$ \nobreakspace1.2} \\
OOS-PST      &   &   &   & \multicolumn{1}{c}{--}  &  \multicolumn{1}{c}{--} & \multicolumn{1}{c}{--}  & \multicolumn{1}{c}{--}  \\
OOS-IST      &   &   &  &   &  {\color{lose} \nobreakspace-2.0 $\pm$ \nobreakspace1.3} &  {\color{lose} \nobreakspace-5.2 $\pm$ \nobreakspace1.2} &  {\color{win} \nobreakspace54.5 $\pm$ \nobreakspace1.1} \\
RM           &   &   &   &   &   &  {\color{lose} -16.6 $\pm$ \nobreakspace1.2} &  {\color{win} \nobreakspace67.8 $\pm$ \nobreakspace1.0} \\
UCT          &   &   &   &   &   &   &  {\color{win} \nobreakspace70.0 $\pm$ \nobreakspace1.0} \\
\bottomrule
\end{tabular}

}
    \end{subtable}%
    \begin{subtable}[t]{.5\linewidth}
        \centering
        \scalebox{1.02}{\begin{tabular}{lrrrrrrl}
\toprule
\textbf{IIGS-13} &                                 MCCR (reset) &                                                   MCCFR &                                           OOS-PST &                                                                   OOS-IST &                                           RM &                                          UCT &                                                     RND \\
\midrule
MCCR (keep)  &  {\color{lose} -18.8 $\pm$ \nobreakspace5.6} &  {\color{win} \nobreakspace12.8 $\pm$ \nobreakspace5.7} & \multicolumn{1}{c}{--}  &  {\color{lose} \nobreakspace-7.3 $\pm$ \nobreakspace5.7} &  {\color{lose} -56.4 $\pm$ \nobreakspace4.7} &  {\color{lose} -69.1 $\pm$ \nobreakspace4.1} &  {\color{win} \nobreakspace43.9 $\pm$ \nobreakspace5.1} \\
MCCR (reset) &   &  {\color{win} \nobreakspace24.4 $\pm$ \nobreakspace5.5} & \multicolumn{1}{c}{--}  &  {\color{win} \nobreakspace\nobreakspace4.9 $\pm$ \nobreakspace2.6} &  {\color{lose} -35.6 $\pm$ \nobreakspace5.3} &  {\color{lose} -56.0 $\pm$ \nobreakspace4.7} &  {\color{win} \nobreakspace55.6 $\pm$ \nobreakspace4.7} \\
MCCFR        &   &   &  \multicolumn{1}{c}{--} &  {\color{lose} -22.8 $\pm$ \nobreakspace5.6} &  {\color{lose} -59.9 $\pm$ \nobreakspace4.6} &  {\color{lose} -75.1 $\pm$ \nobreakspace3.7} &  {\color{win} \nobreakspace37.8 $\pm$ \nobreakspace5.3} \\
OOS-PST      &   &   &   & \multicolumn{1}{c}{--}  & \multicolumn{1}{c}{--}  & \multicolumn{1}{c}{--}  & \multicolumn{1}{c}{--}  \\
OOS-IST      &   &   &   &   &  {\color{lose} -44.4 $\pm$ \nobreakspace5.1} &  {\color{lose} -61.2 $\pm$ \nobreakspace4.5} &  {\color{win} \nobreakspace58.2 $\pm$ \nobreakspace4.6} \\
RM           &   &   &   &   &   &  {\color{lose} -22.8 $\pm$ \nobreakspace5.6} &  {\color{win} \nobreakspace82.3 $\pm$ \nobreakspace3.2} \\
UCT          &   &   &   &   &   &   &  {\color{win} \nobreakspace91.2 $\pm$ \nobreakspace2.3} \\
\bottomrule
\end{tabular}

}
    \end{subtable}\par\bigskip
    \begin{subtable}[t]{.5\linewidth}%
        \centering%
        \scalebox{1.02}{\begin{tabular}{lrrrrrrl}
\toprule
\textbf{LD-116} &                                                   MCCR (reset) &                                                               MCCFR &                                                                   OOS-PST &                                                  OOS-IST &                                                             RM &                                                                       UCT &                                                     RND \\
\midrule
MCCR (keep)  &  {\color{insignlose} \nobreakspace-1.4 $\pm$ \nobreakspace2.0} &  {\color{win} \nobreakspace\nobreakspace9.7 $\pm$ \nobreakspace2.0} &  {\color{win} \nobreakspace\nobreakspace5.2 $\pm$ \nobreakspace2.0} &  {\color{lose} \nobreakspace-4.1 $\pm$ \nobreakspace2.0} &  {\color{win} \nobreakspace2.6 $\pm$ \nobreakspace1.0} &  {\color{win} \nobreakspace\nobreakspace5.6 $\pm$ \nobreakspace2.0} &  {\color{win} \nobreakspace60.9 $\pm$ \nobreakspace1.6} \\
MCCR (reset) &   &  {\color{win} \nobreakspace\nobreakspace11.6 $\pm$ \nobreakspace2.0} &  {\color{win} \nobreakspace\nobreakspace10.1 $\pm$ \nobreakspace2.0} &  {\color{lose} \nobreakspace-3.9 $\pm$ \nobreakspace2.0} &  {\color{lose} \nobreakspace-5.5 $\pm$ \nobreakspace2.0} &  {\color{win} \nobreakspace\nobreakspace5.3 $\pm$ \nobreakspace2.0} &  {\color{win} \nobreakspace60.5 $\pm$ \nobreakspace1.6} \\
MCCFR        &   &   &  {\color{lose} \nobreakspace-6.8 $\pm$ \nobreakspace2.0} &  {\color{lose} \nobreakspace-8.7 $\pm$ \nobreakspace2.0} &  {\color{lose} \nobreakspace-4.7 $\pm$ \nobreakspace2.0} &  {\color{win} \nobreakspace\nobreakspace1.9 $\pm$ \nobreakspace1.0} &  {\color{win} \nobreakspace54.0 $\pm$ \nobreakspace1.7} \\
OOS-PST      &   &   &   &  {\color{lose} \nobreakspace-4.3 $\pm$ \nobreakspace2.0} &  {\color{lose} \nobreakspace-3.0 $\pm$ \nobreakspace2.0} &  {\color{win} \nobreakspace\nobreakspace6.5 $\pm$ \nobreakspace2.0} &  {\color{win} \nobreakspace60.3 $\pm$ \nobreakspace1.6} \\
OOS-IST      &   &   &   &   &  {\color{lose} \nobreakspace-1.7 $\pm$ \nobreakspace1.0} &  {\color{win} \nobreakspace\nobreakspace5.2 $\pm$ \nobreakspace2.0} &  {\color{win} \nobreakspace64.8 $\pm$ \nobreakspace1.6} \\
RM           &   &   &   &   &   &  {\color{win} \nobreakspace\nobreakspace5.2 $\pm$ \nobreakspace2.0} &  {\color{win} \nobreakspace66.1 $\pm$ \nobreakspace1.5} \\
UCT          &   &   &   &   &   &   &  {\color{win} \nobreakspace65.4 $\pm$ \nobreakspace1.5} \\
\bottomrule
\end{tabular}

}
    \end{subtable}%
    \begin{subtable}[t]{.5\linewidth}
        \centering
        \scalebox{1.02}{\begin{tabular}{lrrrrrrl}
\toprule
\textbf{LD-226} &                                 MCCR (reset) &                                                   MCCFR &                                                                   OOS-PST &                                      OOS-IST &                                           RM &                                          UCT &                                                     RND \\
\midrule
MCCR (keep)  &  {\color{win} 6.2 $\pm$ \nobreakspace5.4} &  {\color{win} \nobreakspace45.3 $\pm$ \nobreakspace5.7} &  {\color{win} \nobreakspace\nobreakspace46.1 $\pm$ \nobreakspace5.7} &  {\color{lose} -22.7 $\pm$ \nobreakspace5.9} &  {\color{lose} -33.6 $\pm$ \nobreakspace5.7} &  {\color{lose} -33.4 $\pm$ \nobreakspace5.7} &  {\color{win} \nobreakspace76.2 $\pm$ \nobreakspace4.2} \\
MCCR (reset) &   &  {\color{win} \nobreakspace37.8 $\pm$ \nobreakspace5.4} &  {\color{win} \nobreakspace\nobreakspace44.2 $\pm$ \nobreakspace5.7} &  {\color{lose} -31.2 $\pm$ \nobreakspace5.7} &  {\color{lose} -39.8 $\pm$ \nobreakspace5.4} &  {\color{lose} -44.8 $\pm$ \nobreakspace5.2} &  {\color{win} \nobreakspace81.6 $\pm$ \nobreakspace4.5} \\
MCCFR        &   &   &  {\color{lose} -5.4 $\pm$ \nobreakspace4.6} &  {\color{lose} -55.7 $\pm$ \nobreakspace4.5} &  {\color{lose} -49.3 $\pm$ \nobreakspace4.6} &  {\color{lose} -47.8 $\pm$ \nobreakspace5.1} &  {\color{win} \nobreakspace45.8 $\pm$ \nobreakspace5.2} \\
OOS-PST      &   &   &   &  {\color{lose} -53.6 $\pm$ \nobreakspace5.3} &  {\color{lose} -51.5 $\pm$ \nobreakspace4.9} &  {\color{lose} -46.4 $\pm$ \nobreakspace5.1} &  {\color{win} \nobreakspace49.6 $\pm$ \nobreakspace4.1} \\
OOS-IST      &   &   &   &   &  {\color{lose} -12.0 $\pm$ \nobreakspace5.6} &  {\color{lose} -22.5 $\pm$ \nobreakspace5.6} &  {\color{win} \nobreakspace83.8 $\pm$ \nobreakspace3.8} \\
RM           &   &   &   &   &   &  {\color{lose} -11.6 $\pm$ \nobreakspace5.7} &  {\color{win} \nobreakspace79.7 $\pm$ \nobreakspace3.5} \\
UCT          &   &   &   &   &   &   &  {\color{win} \nobreakspace75.4 $\pm$ \nobreakspace3.8} \\
\bottomrule
\end{tabular}
}
    \end{subtable}\par\bigskip
    \begin{subtable}[t]{.5\linewidth}%
        \centering%
        \scalebox{1.02}{\begin{tabular}{lrrrrrrl}
\toprule
\textbf{GP-3322} &                                                              MCCR (reset) &                                                               MCCFR &                                                                   OOS-PST &                                                                   OOS-IST &                                                       RM &                                                      UCT &                                                                 RND \\
\midrule
MCCR (keep)  &  {\color{insignwin} \nobreakspace\nobreakspace0.2 $\pm$ \nobreakspace0.4} &  {\color{win} \nobreakspace\nobreakspace2.2 $\pm$ \nobreakspace0.5} &  {\color{win} \nobreakspace\nobreakspace1.7 $\pm$ \nobreakspace0.5} &  {\color{win} \nobreakspace\nobreakspace0.4 $\pm$ \nobreakspace0.2} &  {\color{lose} \nobreakspace-0.5 $\pm$ \nobreakspace0.4} &  {\color{lose} \nobreakspace-1.5 $\pm$ \nobreakspace0.4} &  {\color{win} \nobreakspace\nobreakspace5.9 $\pm$ \nobreakspace0.5} \\
MCCR (reset) &   &  {\color{win} \nobreakspace\nobreakspace0.7 $\pm$ \nobreakspace0.4} &  {\color{win} \nobreakspace\nobreakspace0.4 $\pm$ \nobreakspace0.2} &  {\color{lose} \nobreakspace-0.3 $\pm$ \nobreakspace0.2} &  {\color{lose} \nobreakspace-0.5 $\pm$ \nobreakspace0.4} &  {\color{lose} \nobreakspace-1.0 $\pm$ \nobreakspace0.3} &  {\color{win} \nobreakspace\nobreakspace5.5 $\pm$ \nobreakspace0.4} \\
MCCFR        &   &   &  {\color{lose} \nobreakspace-1.9 $\pm$ \nobreakspace0.6} &  {\color{lose} \nobreakspace-2.7 $\pm$ \nobreakspace0.6} &  {\color{lose} \nobreakspace-3.6 $\pm$ \nobreakspace0.5} &  {\color{lose} \nobreakspace-3.3 $\pm$ \nobreakspace0.5} &  {\color{win} \nobreakspace\nobreakspace6.0 $\pm$ \nobreakspace0.6} \\
OOS-PST      &   &   &   &  {\color{lose} \nobreakspace-1.0 $\pm$ \nobreakspace0.9} &  {\color{lose} \nobreakspace-2.1 $\pm$ \nobreakspace0.5} &  {\color{lose} \nobreakspace-2.8 $\pm$ \nobreakspace0.4} &  {\color{win} \nobreakspace\nobreakspace7.4 $\pm$ \nobreakspace0.6} \\
OOS-IST      &   &   &   &   &  {\color{lose} \nobreakspace-1.3 $\pm$ \nobreakspace0.5} &  {\color{lose} \nobreakspace-2.1 $\pm$ \nobreakspace0.4} &  {\color{win} \nobreakspace\nobreakspace7.4 $\pm$ \nobreakspace0.6} \\
RM           &   &   &   &   &   &  {\color{lose} \nobreakspace-1.2 $\pm$ \nobreakspace0.4} &  {\color{win} \nobreakspace\nobreakspace7.8 $\pm$ \nobreakspace0.5} \\
UCT          &   &   &   &   &   &   &  {\color{win} \nobreakspace\nobreakspace6.3 $\pm$ \nobreakspace0.4} \\
\bottomrule
\end{tabular}

}
    \end{subtable}%
    \begin{subtable}[t]{.5\linewidth}
        \centering
        \scalebox{1.02}{\begin{tabular}{lrrrrrrl}
\toprule
\textbf{GP-4644} &                                                        MCCR (reset) &                                                               MCCFR &                                                        OOS-PST &                                                  OOS-IST &                                                             RM &                                                            UCT &                                                                 RND \\
\midrule
MCCR (keep)  &  {\color{win} \nobreakspace\nobreakspace1.6 $\pm$ \nobreakspace1.2} &  {\color{win} \nobreakspace\nobreakspace7.3 $\pm$ \nobreakspace2.6} &  {\color{win} \nobreakspace14.2 $\pm$ \nobreakspace2.5} &  {\color{lose} \nobreakspace-3.4 $\pm$ \nobreakspace2.3} &  {\color{lose} \nobreakspace-4.1 $\pm$ \nobreakspace1.8} &  {\color{lose} \nobreakspace-6.9 $\pm$ \nobreakspace1.5} &  {\color{win} \nobreakspace19.3 $\pm$ \nobreakspace2.6} \\
MCCR (reset) &   &  {\color{win} \nobreakspace9.5 $\pm$ \nobreakspace2.0} &  {\color{win} \nobreakspace11.9 $\pm$ \nobreakspace2.0} &  {\color{lose} \nobreakspace-3.5 $\pm$ \nobreakspace1.8} &  {\color{lose} \nobreakspace-3.0 $\pm$ \nobreakspace1.5} &  {\color{lose} \nobreakspace-2.5 $\pm$ \nobreakspace1.3} &  {\color{win} \nobreakspace15.8 $\pm$ \nobreakspace2.2} \\
MCCFR        &   &   &  {\color{lose} \nobreakspace-8.7 $\pm$ \nobreakspace3.1} &  {\color{lose} -13.3 $\pm$ \nobreakspace2.9} &  {\color{lose} \nobreakspace-9.6 $\pm$ \nobreakspace2.3} &  {\color{lose} \nobreakspace-6.8 $\pm$ \nobreakspace1.9} &  {\color{win} \nobreakspace12.0 $\pm$ \nobreakspace3.0} \\
OOS-PST      &   &   &   &  {\color{lose} \nobreakspace-8.1 $\pm$ \nobreakspace3.0} &  {\color{lose} \nobreakspace-8.9 $\pm$ \nobreakspace2.4} &  {\color{lose} \nobreakspace-5.0 $\pm$ \nobreakspace2.0} &  {\color{win} \nobreakspace11.8 $\pm$ \nobreakspace3.1} \\
OOS-IST      &   &   &   &   &  {\color{lose} \nobreakspace-2.1 $\pm$ \nobreakspace1.8} &  {\color{lose} \nobreakspace-1.6 $\pm$ \nobreakspace1.2} &  {\color{win} \nobreakspace20.4 $\pm$ \nobreakspace2.9} \\
RM           &   &   &   &   &   &  {\color{insignlose} \nobreakspace-0.3 $\pm$ \nobreakspace1.1} &  {\color{win} \nobreakspace20.5 $\pm$ \nobreakspace2.3} \\
UCT          &   &   &   &   &   &   &  {\color{win} \nobreakspace17.6 $\pm$ \nobreakspace2.0} \\
\bottomrule
\end{tabular}

}
    \end{subtable}\par\bigskip
    \begin{subtable}[t]{.5\linewidth}%
        \centering%
        \begin{tabular}{l}

        \end{tabular}
    \end{subtable}%
    \begin{subtable}[t]{.5\linewidth}
        \centering
        \scalebox{1.02}{\begin{tabular}{lrrrrrrl}
\toprule
\textbf{PTTT} &                                            MCCR (reset) &                                                          MCCFR &                                                                   OOS-PST &                                                        OOS-IST &                                                             RM &                                                                 UCT &                                                     RND \\
\midrule
MCCR (keep)  &  {\color{win} \nobreakspace17.7 $\pm$ \nobreakspace3.8} &  {\color{lose} \nobreakspace-1.1 $\pm$ \nobreakspace0.9} &  {\color{lose} \nobreakspace-1.8 $\pm$ \nobreakspace1.6} &  {\color{lose} \nobreakspace-5.0 $\pm$ \nobreakspace3.7} &  {\color{lose} \nobreakspace-6.9 $\pm$ \nobreakspace3.7} &  {\color{lose} \nobreakspace-6.2 $\pm$ \nobreakspace3.7} &  {\color{win} \nobreakspace25.5 $\pm$ \nobreakspace3.8} \\
MCCR (reset) &   &  {\color{lose} \nobreakspace-6.2 $\pm$ \nobreakspace3.8} &  {\color{lose} -9.8 $\pm$ \nobreakspace3.7} &  {\color{lose} -14.6 $\pm$ \nobreakspace3.6} &  {\color{lose} -20.9 $\pm$ \nobreakspace3.6} &  {\color{lose} -14.3 $\pm$ \nobreakspace3.7} &  {\color{win} \nobreakspace21.6 $\pm$ \nobreakspace3.7} \\
MCCFR        &   &   &  {\color{insignwin} \nobreakspace\nobreakspace0.1 $\pm$ \nobreakspace3.7} &  {\color{lose} \nobreakspace-2.1 $\pm$ \nobreakspace1.5} &  {\color{lose} \nobreakspace-5.2 $\pm$ \nobreakspace3.7} &  {\color{lose} \nobreakspace-4.0 $\pm$ \nobreakspace3.6} &  {\color{win} \nobreakspace27.9 $\pm$ \nobreakspace3.7} \\
OOS-PST      &   &   &   &  {\color{lose} \nobreakspace-5.6 $\pm$ \nobreakspace3.7} &  {\color{lose} \nobreakspace-5.7 $\pm$ \nobreakspace3.7} &  {\color{lose} \nobreakspace-5.2 $\pm$ \nobreakspace3.7} &  {\color{win} \nobreakspace29.4 $\pm$ \nobreakspace3.7} \\
OOS-IST      &   &   &   &   &  {\color{lose} \nobreakspace-3.5 $\pm$ \nobreakspace3.2} &  {\color{lose} \nobreakspace-3.9 $\pm$ \nobreakspace3.6} &  {\color{win} \nobreakspace35.1 $\pm$ \nobreakspace3.6} \\
RM           &   &   &   &   &   &  {\color{win} \nobreakspace\nobreakspace5.6 $\pm$ \nobreakspace3.6} &  {\color{win} \nobreakspace51.3 $\pm$ \nobreakspace3.3} \\
UCT          &   &   &   &   &   &   &  {\color{win} \nobreakspace52.6 $\pm$ \nobreakspace3.3} \\
\bottomrule
\end{tabular}

}
    \end{subtable}
    \rule{0pt}{3mm}
    \caption{Head-to-head performance. Positive numbers mean that the row algorithm is winning against the column algorithm by the given percentage of the maximum payoff in the domain. Gray numbers indicate the winner isn't statistically significant.}
    \label{tab:matches}
\end{table*}
\normalsize

\section{Conclusion}

We propose a~generalization of Continual Resolving from poker~\cite{DeepStack} to other extensive-form games. We show that the structure of the public tree may be more complex in general, and propose an extended version of the resolving gadget necessary to handle this complexity.
Furthermore, both players may play in the same public state (possibly multiple times), and we extend the definition of Continual Resolving to allow this case as well. We present a completely domain-independent version of the algorithm that can be applied to any EFG, is sufficiently robust to use variable resolving schemes, and can be coupled with different resolving games and algorithms (including classical CFR, depth-limited search utilizing neural networks, or other domain-specific heuristics). We show that the existing theory naturally translates to this generalized case.


We further introduce Monte Carlo CR as a specific instance of this abstract algorithm that uses \ac{MCCFR} as a resolver. It allows deploying continual resolving on any domain, without the need for expensive construction of evaluation functions.
\ac{MCCR} is theoretically sound as demonstrated by Theorem~\ref{thm:mccr}, constitutes an
improvement over \ac{MCCFR} in the online setting in terms head-to-head performance,
and doesn't have the restrictive memory requirements of \ac{OOS}. The experimental evaluation shows that \ac{MCCR} is very sensitive to the quality of its counterfactual value estimates. With good estimates, its worst-case performance (i.e. exploitability) improves faster than that of \ac{OOS}. In head-to-head matches \ac{MCCR} plays similarly to \ac{OOS}, but  it only outperforms \ac{IS-MCTS} in one of the smaller tested domains. Note however that the lack of theoretical guarantees of \ac{IS-MCTS} often translates into severe exploitability in practice~\cite{OOS}, and this cannot be alleviated by increasing \ac{IS-MCTS}'s computational resources~\cite{ISMCTS}.
In domains where \ac{MCCR}'s counterfactual value estimates are less precise, its exploitability still converges to zero, but at a slower rate than \ac{OOS}, and its head-to-head performance is noticeably weaker than that of both \ac{OOS} and \ac{IS-MCTS}.

In the future work, the quality of \ac{MCCR}'s estimates might be improved by variance reduction~\cite{var_reduction}, exploring ways of improving these estimates over the course of the game, or by finding an alternative source from which they can be obtained.
We also plan to test the hypothesis that there are classes of games where MCCR performs much better than the~competing algorithms (in particular, we suspect this might be true for small variants of turn-based computer games such as Heroes of Might \& Magic or Civilization).

\section*{Acknowledgments}
Access to computing and storage facilities owned by parties and projects contributing to the National Grid Infrastructure MetaCentrum provided under the program "Projects of Large Research, Development, and Innovations Infrastructures" (CESNET LM2015042), is greatly appreciated.
This work was supported by Czech science foundation grant no. 18-27483Y.



\bibliographystyle{ACM-Reference-Format}  
\newpage
\balance
\bibliography{refs}  

\clearpage
\newpage
\section*{Appendix}
\appendix
\printacronyms[include-classes=abbrev,name=Abbreviations]

\section{The Proof of Theorem~\ref{thm:mccr}}\label{sec:proofs}

We now formalize the guarantees available for different ingredients of our continual resolving algorithm, and put them together to prove Theorem~\ref{thm:mccr}.

%

\subsection{Monte Carlo CFR}

The basic tool in our algorithm is the outcome sampling (\ac{OS}) variant of \ac{MCCFR}.
In the following text, $p \in (0,1]$ and $\delta >0$ will be fixed,
and we shall assume that the \ac{OS}'s sampling scheme is such that for each $z\in \mc Z$, the probability of sampling $z$ is either 0 or higher than $\delta$.

In~\cite[Theorem 4]{lanctot_thesis}, it is proven that the \ac{OS}'s average regret decreases with increasing number of iterations. This translates to the following exploitability bound, where
$\Delta_{u,i} := \max_{z1,z_2} |u_i(z_1) - u_i(z_2)|$ is the maximum difference between utilities and
$A_i := \max_{\mc H_i} |\mc A(h)|$ is the player $i$ branching factor.

\begin{lemma}[\ac{MCCFR} exploitability bound]\label{lem:MCCFR_expl}
Let $\bar \sigma^T$ be the average strategy produced by applying $T$ iterations of \ac{OS} to some game $G$. Then with probability at least $1-p$, we have
\begin{equation} \label{eq:MCCFR_regret}
\bar R^T_i \leq 
\left( \sqrt{\frac{2}{p}} + 1 \right)|\mc I_i|
\frac{\Delta_{u,i}\sqrt{A_i}}{\delta} \cdot \frac{1}{\sqrt{T}} .
\end{equation}
\end{lemma}

\begin{proof}
The exploitability $\expl_i(\bar \sigma^T)$ of the average strategy $\bar \sigma^T$ is equal to the average regret $\bar R^T_i$.
By Theorem 4 from~\cite{lanctot_thesis}, after $T$ iterations of \ac{OS} we have
\[ \bar R^T_i \leq 
\left( \frac{\sqrt{2|\mc I_i||\mc B_i|}}{\sqrt{p}} + M_i \right)
\frac{\Delta_{u,i}\sqrt{A_i}}{\delta} \cdot \frac{1}{\sqrt{T}}
\]
with probability at least $1-p$, where
$M_i, |\mc B_i| \leq |\mc I_i|$ are some domain specific constants.
The regret can then be bounded as
\[ \bar R^T_i \leq 
\left( \sqrt{\frac{2}{p}} + 1 \right)|\mc I_i|
\frac{\Delta_{u,i}\sqrt{A_i}}{\delta} \cdot \frac{1}{\sqrt{T}},
\]
which concludes the proof.
\end{proof}

\begin{lemma}[\ac{MCCFR} value approximation bound]\label{lem:MCCFR_values}
Let $S\in \mc S$. Under the assumptions of Lemma~\ref{lem:MCCFR_expl}, we further have 
\begin{align*}
\sum_{I\in S(2)} \left|
v_2^{\bar \sigma^T}(I) - v_2^{\bar\sigma^T_1 , \textrm{CBR}(\bar\sigma^T_1) }(I)
\right| \leq \\
\left( \sqrt{\frac{2}{p}} + 1 \right)|\mc I_i|
\frac{\Delta_{u,i}\sqrt{A_i}}{\delta} \cdot \frac{1}{\sqrt{T}}.
\end{align*}
\end{lemma}

\begin{proof}
Consider the \emph{full counterfactual regret} of player $2$'s average strategy, defined in~\cite[Appendix A.1]{CFR}:
\begin{equation}\label{eq:full_regret}
R^T_{2,\textnormal{full}}(I) :=
\frac{1}{T} \max_{\sigma'_2 \in \Sigma_2} \sum_{t=1}^T \left(
v_2^{\sigma^t|_{D(I) \leftarrow \sigma'_2}}(I) - v_2^{\sigma^t}(I) 
\right) ,
\end{equation}
where $D(I)\subset \mc I_2$ contains $I$ and all its descendants.
Since $\textrm{CBR}_2(\bar \sigma^T)$ is one of the strategies $\sigma'_2$ which maximize the sum in \eqref{eq:full_regret}, we have
\begin{equation} \label{eq:regret_and_value_approx}
R^T_{2,\textnormal{full}}(I) = v_2^{\bar\sigma^T_1 , \textrm{CBR}(\bar\sigma_1) }(I) - v_2^{\bar \sigma^T}(I) .
\end{equation}

Consider now any $S \in \mc S$. By~\cite[Lemma\,6]{CFR}, we have
\begin{align*}
& \sum_{I\in S(2)} \left|
v_2^{\bar \sigma^T}(I) - v_2^{\bar\sigma^T_1, \textrm{CBR}(\bar\sigma_1) }(I)
\right| \overset{\eqref{eq:regret_and_value_approx} }{\leq} \!\!\!
\sum_{I\in S(2)} \!\! R^T_{2,\textnormal{full}}(I) \\
& \leq \sum_{I\in S(2)} \sum_{J\in D(I)}
 R^{T,+}_{i,\textnormal{imm}}(J) \leq
\sum_{J \in \mc I_2} R^{T,+}_{i,\textnormal{imm}}(J) .
\end{align*}
The proof is now complete, because the proof of~\cite[Theorem 4]{lanctot_thesis}, which we used to prove Lemma~\ref{lem:MCCFR_expl}, actually shows that the sum $\sum_{J \in \mc I_2} R^{T,+}_{i,\textnormal{imm}}(J)$ is bounded by the right-hand side of \eqref{eq:MCCFR_regret}.
\end{proof}

\subsection{Gadget Game Properties}

The following result is a part of why resolving gadget games are useful, and the reason behind multiplying all utilities in $G\left< S, \sigma, v_2^\sigma \right>$ by $\pi_{-2}^\sigma(S)$.

\begin{lemma}[Gadget game preserves opponent's values]\label{lem:gadget_values}
Let $S\in \mc S$, $\sigma \in \Sigma$ and let $\tilde \rho$ be any strategy in the resolving game $G\left< S, \sigma, \tilde v\right>$ (where $\tilde v$ is arbitrary).
Denote $\sigma^\textnormal{new} := \sigma |_{G(S) \leftarrow \tilde \rho}$.
Then for each $I \in \mc I_2^{\textnormal{aug}}(G(S))$, we have
$v_2^{\sigma^\textnormal{new}}(I) = \tilde v_2^{\tilde \rho}(I)$.
\end{lemma}

\noindent Note that the conclusion does \emph{not} hold for counterfactual values of the (resolving) player 1! (This is can be easily verified by hand in any simple game such as matching pennies.)

\begin{proof}
In the setting of the lemma, it suffices to show that $\tilde{v}_2^{\tilde\rho}(h)$ is equal to $v_2^{\sigma^{\textnormal{new}}}(h)$ for every $h\in G(S)$.
Let $h\in G(S)$ and denote by $g$ the prefix of $h$ that belongs to $\hat S$.

Recall that $\tilde g$ is the parent of $g$ in the resolving game, and that the reach probability $\tilde \pi_{-2}^{\tilde \mu}(\tilde g)$ of $\tilde g$ in the resolving game (for any strategy $\tilde \mu$) is equal to $\pi^\sigma_{-2}(g) / \pi_{-2}^\sigma (S)$.
Since $\tilde u_2 (z) = u_2(z) \pi^\sigma_{-2}(S)$ for any $h\sqsubset z \in \mc Z$, the definition of $\sigma^{\textnormal{new}}$, gives $\tilde u_2^{\tilde\rho}(h) = u_2^{\sigma^{\textnormal{new}}}(h) \pi_{-2}^\sigma(S)$.
The following series of identities then completes the proof:

\begin{align*}
\tilde{v}_2^{\tilde\rho}(h) &
	= \sum_{h\sqsubset z \in \mc Z} \tilde \pi_{-2}^{\tilde\rho}(h) {\tilde\pi}^{\tilde\rho}(z|h) \tilde{u}_2(z) \\
& = \tilde \pi_{-2}^{\tilde\rho}(h) \tilde u_2^{\tilde\rho}(h) \\
& = \tilde \pi^{\tilde\rho}_{-2}(\tilde g) \cdot \tilde \pi_{-2}(g | \tilde g) \cdot \tilde \pi_{-2}^{\tilde\rho}(h|g) \cdot \tilde u_2^{\tilde\rho}(h) \\
& = \tilde \pi_c^{\tilde\rho}(\tilde g) \cdot 1 \cdot \tilde \pi_{-2}^{\tilde\rho}(h|g) \cdot \tilde u_2^{\tilde\rho}(h) \\
& = \frac{ \pi^\sigma_{-2}(g) }{ \pi_{-2}^\sigma (S) } \cdot 1 \cdot \pi_{-2}(h|g) \cdot u_2^{\sigma^{\textnormal{new}}}(h) \pi_{-2}^\sigma(S) \\
& = \pi^\sigma_{-2}(h) u_2^{\sigma^{\textnormal{new}}}(h) 
 = \pi^{\sigma^{\textnormal{new}}}_{-2}(h) u_2^{\sigma^{\textnormal{new}}}(h)
 = v_2^{\sigma^{\textnormal{new}}}(h) .
\end{align*}
\end{proof}

\begin{corollary}[Gadget game preserves value approximation]\label{lem:gadget_BR}
Under the assumptions of Lemma~\ref{lem:gadget_values}, we have the following for each $I \in \mc I_2^{\textnormal{aug}}(G(S))$:
\begin{align*}
& \left|
v_2^{\sigma^\textnormal{new}}(I) - 
v_2^{\sigma^\textnormal{new}_1, \textrm{CBR}(\sigma^\textnormal{new}_1)}(I)
\right| = \\
& \left|
\tilde{v}_2^{\tilde \rho}(I) -
\tilde{v}_2^{\tilde\rho_1, \widetilde{\textrm{CBR}}(\tilde\rho_1)}(I)
\right| .
\end{align*}
\end{corollary}

\begin{proof}
By Lemma~\ref{lem:gadget_values}, we have $v_2^{\sigma^\textnormal{new}}(I) = \tilde{v}_2^{\tilde \rho}(I)$.
We claim that $\widetilde{\textrm{CBR}}(\tilde\rho_1)$ coincides with $\textrm{CBR}(\sigma^\textnormal{new}_1)$ on $G(S)$, which (again by Lemma~\ref{lem:gadget_values}) implies the equality of the corresponding counterfactual values, and hence the conclusion of the corollary.
To see that the claim is true, recall that the counterfactual best-response is constructed by the standard backwards-induction in the corresponding game tree (player 2 always picks the action with highest $v^{(\cdot)}_2(I,a)$). 
We need this backwards induction to return the same strategy independently of whether $G(S)$ is viewed as a subgame of $G$ or $G\left< S, \sigma, v_2^\sigma \right>$.
But this is trivial, since both the structure of $G(S)$ and counterfactual values of player 2 are the same in both games.
This concludes the proof.
\end{proof}

\subsection{Resolving}

The second useful property of resolving games is that if we start with an approximately optimal strategy, and find an approximately optimal strategy in the resolving game, then the new strategy will again be approximately optimal.
This is formalized by the following immediate corollary of Lemma 24 and 25 from~\cite{DeepStack}:

\begin{lemma}[Resolved strategy performs well]\label{lem:resolving_lemma}
Let $S\in \mc S$, $\sigma \in \Sigma$ and denote by $\tilde \rho$ a strategy in the resolving game $\widetilde G\left( S, \sigma, v_2^\sigma\right)$.
Then the exploitability $\expl_1(\sigma^\textnormal{new})$ of the strategy
$\sigma^\textnormal{new} := \sigma |_{G(S) \leftarrow \tilde \rho}$ is no greater than
\[ \expl_1(\sigma) + \sum_{I \in S(2)}
\left| v_2^{\sigma_1,\textrm{CBR}(\sigma_1)}(I) - v_2^\sigma(I) \right| + \widetilde \expl_1( \tilde \rho)
.\]
\end{lemma}



\subsection{Monte Carlo Continual Resolving}

Suppose that \ac{MCCR} is run with parameters $T_0$ and $T_R$. For a game requires at most $N$ resolving steps, we then have the following theoretical guarantee:

\MCCRbound*

\begin{proof}
Without loss of generality, we assume that \ac{MCCR} acts as player 1.
We prove the theorem by induction. The initial step of the proof is different from the induction steps, and goes as follows.
Let $S_0$ be the root of the game and denote by $\sigma^0$ the strategy obtained by applying $T_0$ iterations of \ac{MCCFR} to $G$.
Denote by $\epsilon^E_0$ the upper bound on $\expl_1(\sigma^0)$ obtained from Lemma~\ref{lem:MCCFR_expl}.
If $S_1$ is the first encountered public state where player 1 acts then by Lemma~\ref{lem:MCCFR_values} $\sum_{I\in S_1(2))} \left|v_2^{\sigma^0}(I) - v_2^{\sigma^0_1 , \textrm{CBR}(\sigma_1^0) }(I)\right|$ is bounded by some $\epsilon^A_0$.
This concludes the initial step.

For the induction step, suppose that $n\geq 1$ and player 1 has already acted $(n-1)$-times according to some strategy $\sigma^{n-1}$ with $\expl_1(\sigma^{n-1})\leq \epsilon^E_{n-1}$, and is now in a public state $S_n$ where he needs to act again. Moreover, suppose that there is some $\epsilon^A_{n-1}\geq 0$ s.t. 
\[ \sum_{S_n(2)} \left|
v_2^{\sigma^{n-1}}(I) - v_2^{\sigma^{n-1}_1 , \textrm{CBR}(\sigma^{n-1}_1) }(I)
\right| \leq \epsilon^A_{n-1} .\]

We then obtain some strategy $\tilde \rho^n$ by resolving the game $\widetilde G_{n} := G\left<S_{n}, \sigma^n,v_2^{\sigma^n}\right>$ by $T_R$ iterations of \ac{MCCFR}.
By Lemma~\ref{lem:MCCFR_expl}, the exploitability $\widetilde \expl_1(\tilde \rho^n)$ in $\widetilde G_{n}$ is bounded by some $\epsilon^R_n$.

We choose our next action according to the strategy
\[ \sigma^n := \sigma^{n-1} |_{G(S_n) \leftarrow \tilde \rho^n} .\]
By Lemma~\ref{lem:resolving_lemma}, the exploitability $\expl_1(\sigma^n)$ is bounded by $\epsilon^E_{n-1} + \epsilon^A_{n-1} + \epsilon^R_n =: \epsilon^E_n$.
The game then progresses until it either ends without player 1 acting again, or reaches a new public state $S_{n+1}$ where player 1 acts.

If such $S_{n+1}$ is reached, then by Lemma~\ref{lem:MCCFR_values}, the value approximation error
$\sum_{S_{n+1}(2)} \left|
\tilde v_2^{\tilde \rho^n}(I) - \tilde v_2^{\tilde \rho^n_1 , \textrm{CBR}(\tilde \rho^n_1) }(I)
\right|$
\emph{in the resolving gadget game} is bounded by some $\epsilon^A_{n+1}$. By Lemma~\ref{lem:gadget_values}, this sum is equal to the value approximation error
$$\sum_{S_{n+1}(2)} \left|
v_2^{\sigma^{n}}(I) - v_2^{\sigma^{n}_1 , \textrm{CBR}(\sigma^{n}_1) }(I)
\right|$$
in the original game.
This concludes the inductive step.

Eventually, the game reaches a terminal state after visiting some sequence $S_1,\dots,S_{n}$ of public states where player 1 acted by using the strategy $\sigma := \sigma^N$. We now calculate the exploitability of $\sigma$.
It follows from the induction that
\[ \expl_1(\sigma) \leq \epsilon^E_N = \epsilon^E_0 + \epsilon^A_0 + \epsilon^R_1 + \epsilon^A_1 + \epsilon^R_2 + \dots + \epsilon^A_{N-1} + \epsilon^R_N .\]
To emphasize which variables come from using $T_0$ iterations of \ac{MCCFR} in the original game and which come from applying $T_R$ iterations of \ac{MCCFR} to the resolving game, we set $\tilde \epsilon^R_n := \epsilon^R_n$ and $\tilde \epsilon^A_n := \epsilon^A_n$ for $n\geq 1$. We can then write
\begin{equation}\label{eq:main_thm_precise}
\expl_1(\sigma) \leq \epsilon^E_N = \epsilon^E_0 + \epsilon^A_0 + \sum_{n=1}^{N-1} (\tilde \epsilon^R_n + \tilde \epsilon^A_n ) + \tilde \epsilon^R_N.
\end{equation}
Since the bound from Lemma~\ref{lem:MCCFR_expl} is strictly higher than the one from Lemma~\ref{lem:MCCFR_values}, we have $\epsilon^A_0 \leq \epsilon^E_0$ and $\tilde \epsilon^A_n \leq \tilde \epsilon^R_n$.
Moreover, we have $G(S_1) \supset G(S_2) \supset \dots G(S_N)$, which means that $\tilde \epsilon^R_1 \geq \tilde \epsilon^R_2 \geq \dots \tilde \epsilon^R_N$.
It follows that $\expl_1(\sigma) \leq 2 \epsilon^E_0 + (2N-1)\tilde \epsilon^R_1$. Finally, we clearly have $N\leq D_1$, where $D_1$ is the ``player 1 depth of the public tree of $G$''. This implies that
\begin{equation}\label{eq:main_thm_nice}
\expl_1(\sigma) \leq 2 \epsilon^E_0 + (2D_1 - 1) \tilde \epsilon^R_1 .
\end{equation}

Plugging in the specific numbers from Lemma~\ref{lem:MCCFR_expl} for $\epsilon^E_0$ and $\tilde \epsilon^R_1$ gives the exact bound, and noticing that we have used the lemma $(D_1+1)$-times implies that the result holds with probability $(1-p)^{D_1+1}$.
\end{proof}

Note that a tighter bound could be obtained if we were more careful and plugged in the specific bounds from Lemma~\ref{lem:MCCFR_expl} and \ref{lem:MCCFR_values} into \eqref{eq:main_thm_precise}, as opposed to using \eqref{eq:main_thm_nice}. Depending on the specific domain, this would yield something smaller than the current bound, but higher than $\epsilon^E_0 + \tilde \epsilon^R_1$.

\section{Computing Counterfactual Values Online}\label{sec:cf_values}

For the purposes of \ac{MCCR}, we require that our solver (\ac{MCCFR}) also returns the counterfactual values of the average strategy.
The straightforward way of ensuring this is to simply calculate the counterfactual values once the algorithm has finished running. However, this might be computationally intractable in larger games, since it potentially requires traversing the whole game tree.
One straightforward way of fixing this issue is to replace this exact computation by a sampling-based evaluation of $\bar \sigma^T$. With a sufficient number of samples, the estimate will be reasonably close to the actual value.

In practice, this is most often solved as follows.
During the normal run of \ac{MCCFR}, we additionally compute the opponent's \emph{sampled counterfactual values}
\[ \tilde v_2^{\sigma^t}(I) :=
\frac{1}{\pi^{\sigma'}(z)}  \pi^{\sigma^t}_{-2}(h)  \pi^{\sigma^t} \! (z|h)  u_2(z)
.\]
Once the $T$ iterations are complete, the counterfactual values of $\bar \sigma^T$ are estimated by $\tilde v(I) := \frac{1}{T}\sum \tilde v_2^{\sigma^t}(I)$.
While this arithmetical average is the standard way of estimating $v_2^{\bar \sigma^T}(I)$, it is also natural to consider alternatives where the uniform weights are replaced by either $\pi^{\sigma^t}(I)$ or $\pi_2^{\sigma^t}(I)$.
In principle, it is possible for all of these weighting schemes to fail (see the counterexample in Section~\ref{sec:cfv_counter_ex}).
We experimentally show that all of these possibilities produce good estimates of $v_2^{\bar \sigma^T}$ in many practical scenarios, see~Figure~\ref{fig:cfv_comparison_domains}.
Even when this isn't the case, one straightforward way to fix this issue is to designate a part of the computation budget to a sampling-based evaluation of $\bar \sigma^T$.
Alternatively, in Lemma~\ref{lem:u_of_avg} we present a method inspired by lazy-weighted averaging from~\cite{lanctot_thesis} that allows for computing unbiased estimates of $v_2^{\sigma^T}$ on the fly during \ac{MCCFR}.

In the main text, we have assumed that the exact counterfactual values are calculated, and thus that $\tilde v(I) = v_2^{\bar \sigma^T}(I)$.
Note that this assumption is made to simplify the theoretical analysis -- in practice, the difference between the two terms can be incorporated into Theorem~\ref{thm:mccr} (by adding the corresponding term into Lemma~\ref{lem:resolving_lemma}).

\begin{figure}[t]
	\includegraphics[width=1\linewidth]{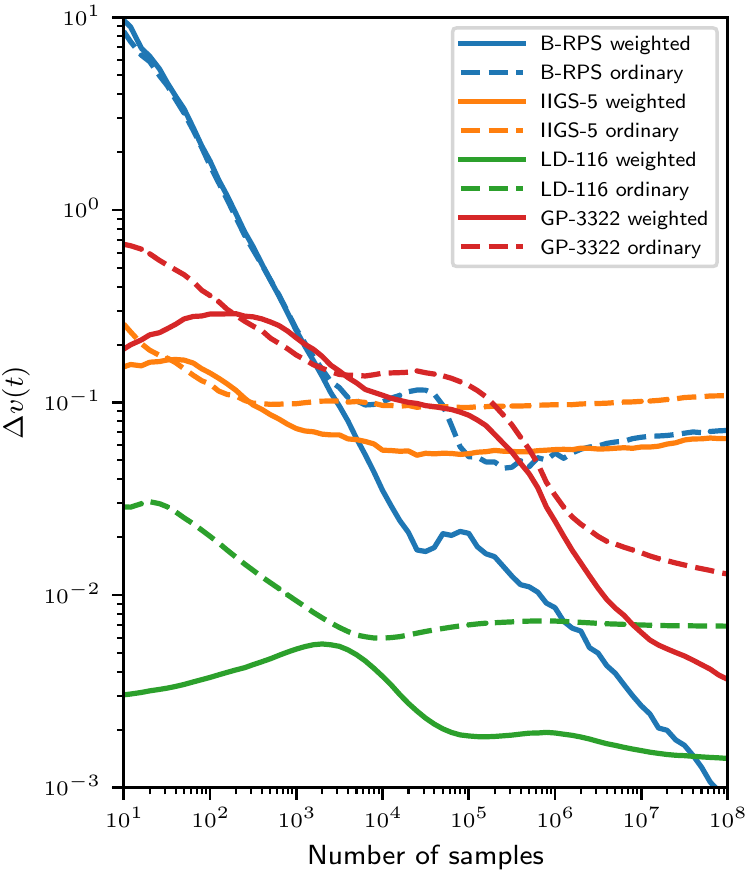}
	\caption{An extension of experiment~\ref{sec:experiment-cfv-averaging}~and~Figure~\ref{fig:expl} (left) for small domains.
	We calculate the exact values $u(h)$, and compute the absolute differences of each weighing sampling scheme ($\tilde{u}^1(h)$ and $\tilde{u}^2(h)$) to exact values. Those differences are then averaged across information sets and seeds.
	In each domain, the weighted averages (solid) have smaller error than ordinary averages (dashed) and are thus better approximations of exact values.}
	\label{fig:cfv_comparison_domains}
\end{figure}

\subsection{The Counterexample}\label{sec:cfv_counter_ex}

In this section we show that no weighting scheme can be used as a universal method for computing the counterfactual values of the average strategy on the fly.
We then derive an alternative formula for $v_i^{\bar \sigma^T}$ and show that it can be used to calculate unbiased estimates of $v_i^{\bar \sigma^T}$ in \ac{MCCFR}.
Note that this problem is not specific to \ac{MCCFR}, but also occurs in CFR (although it is not so pressing there, since CFR's iterations are already so costly that the computation of exact counterfactual values of $\bar \sigma^T$ is not a major expense).
But since these issues already arise in \ac{CFR}, we will work in this simpler setting.

Suppose we have a history $h\in \mc H$, strategies $\sigma^1, \dots, \sigma^T$ and the average strategy $\bar \sigma^T$ defined as
\begin{equation}\label{eq:avg_strat}
\bar \sigma^T (I) := \sum_t \pi_i^{\sigma^t}\!\!(I) \, \sigma^t(I) \, / \, \sigma_t \pi_i^{\sigma^t}\!\!(I)
\end{equation}
for $I\in \mc I_i$.
First, note that we can easily calculate $\pi^{\bar \sigma^T}_{-i}(h)$. Since $v_i^{\bar \sigma^T}(h) = \pi^{\bar \sigma^T}_{-i}(h) u_i^{\bar \sigma^T}(h)$, an equivalent problem is that of calculating the expected utility of the average strategy at $h$ on the fly, i.e. by using $u_i^{\sigma^t}(h)$ and possibly some extra variables, but without having to  traverse the whole tree below $h$.

Looking at the definition of the average strategy, the most natural candidates for an estimate of $u_i^{\bar \sigma^T}(h)$ are the following weighted averages of $u_i^{\sigma^t}(h)$:
\begin{align*}
\tilde u^1(h) & := \sum_t u_i^{\sigma^t}\!\!(h) \, / \, T \\
\tilde u^2(h) & := \sum_t \pi_i^{\sigma^t}\!\!(h) \, u_i^{\sigma^t}\!\!(h) \, / \, \sum_t \pi_i^{\sigma^t}\!\!(g) \\
\tilde u^3(h) & := \sum_t \pi^{\sigma^t}\!\!(h) \, u_i^{\sigma^t}\!\!(h) \, / \, \sum_t \pi^{\sigma^t}\!\!(g) .
\end{align*}

\def\averageStrategyFail{
\tikzset{
	level 1/.style = {level distance = 1.5\nodesize, sibling distance=1.5\nodesize},
	level 2/.style = {level distance = 1.5\nodesize, sibling distance=1.5\nodesize},
	level 3/.style = {level distance = 1.5\nodesize, sibling distance=1.5\nodesize},
	level 4/.style = {level distance = 1.5\nodesize, sibling distance=1.5\nodesize},
	level 5/.style = {level distance = 1.5\nodesize, sibling distance=1.5\nodesize},
}
\node(1)[pl1]{$h_0$}
	child[grow=up]{node[terminal]{0}}
	child[grow=right]{node(2)[pl2]{$h_1$}
		child[grow=up]{node[terminal]{0}}
		child[grow=right]{node(3)[pl1]{$h_2$}
			child[grow=up]{node[terminal]{0}}
			child[grow=right]{node(4)[pl2]{$h_3$}
				child[grow=up]{node[terminal]{0}}
				child[grow=right]{node[terminal,label=above:{z}]{1}}
			}
		}
	}
;
}
\begin{figure}[h]
\centering 
\begin{tikzpicture}
\averageStrategyFail
\end{tikzpicture}
\caption{A domain where weighting schemes for $v_2^{\bar \sigma^T}$ fail.}
\label{fig:avg_str_fail}
\end{figure}
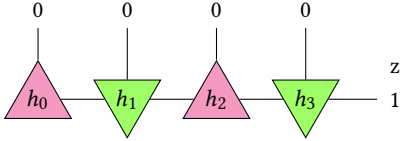

\begin{example}[No weighting scheme works]
Each of the estimates $\tilde u^j(h)$, $j=1,2,3$, can fail to be equal to   $u_i^{\bar \sigma^T}$. Yet worse, no similar works reliably for every sequence $(\sigma^t)_t$.
\end{example}

\noindent Consider the game from Figure~\ref{fig:avg_str_fail} with $T=2$, where under $\sigma^1$, each player always goes right (R) and under $\sigma^2$ the probabilities of going right are $\frac{1}{2}, \frac{1}{3}, \frac{1}{5}$ and $\frac{1}{7}$ at $h_0$, $h_1$, $h_2$ and $h_3$ respectively.
A straightforward application of \eqref{eq:avg_strat} shows that the probabilities of R under $\bar \sigma^2$ are $\frac 3 4$, $\frac 2 3$, $\frac{11}{15}$ and $\frac{11}{14}$ and hence $u_1^{\bar \sigma^2}(h_2) = \frac{11}{15}\cdot \frac{11}{14} = \frac{121}{210}$. On the other hand, we have $\tilde u^1(h_2)=\frac{108}{210}$, $\tilde u^2(h_2)=\frac{142}{210}$ and $\tilde u^3(h_2)\doteq\frac{181}{210}$.

To prove the ``yet worse'' part, consider also the sequence of strategies $\nu^1$, $\nu^2$, where $\nu^1 = \sigma^1$ and under $\nu^2$, the probabilities of going right are $\frac{1}{2}, \frac{1}{3}, 1$ and $\frac{1}{5}\cdot\frac{1}{7}$ at $h_0$, $h_1$, $h_2$ and $h_3$ respectively.
The probabilities of R under $\bar \nu^2$ are $\frac 3 4$, $\frac 2 3$, $1$ and $\frac{27}{35}$ and hence $u_1^{\bar \nu^2}(h_2) = \frac{27}{35} = \frac{162}{210} \neq \frac{121}{210} = u_1^{\bar \sigma^2}(h_2)$.
However, the strategies $\sigma^t$ and $\nu^t$ coincide between the root and $h_2$ for each $t$, and so do the utilities $u^{\sigma^t}(h_0)$ and $u^{\nu^t}(h_0)$.
Necessarily, any weighting scheme of the form
\[ \tilde u(h) := \sum_t w_t(h) u_i^{\sigma^t}\!\!(h) \, / \, \sum_t w_t(h)  \]
where $w_t(h) \in \R$ only depends on the strategy $\sigma^t$ between the root and $h$, yields the same estimate for $(\sigma^t)_{t=1,2}$ and for $(\nu^t)_{t=1,2}$.
As a consequence, any such weighting scheme will be wrong in at least one of these cases.

\subsection{An Alternative Formula for the Utility of \texorpdfstring{$\bar \sigma^T$}{the average strategy}}

We proceed in three steps. First, we derive an alternative formula for $u^{\bar \sigma^T}(h)$ that uses the cumulative reach probabilities
\[ \textnormal{crp}_i^t(z) := \pi_i^{\sigma^1}(z) + \dots + \pi_i^{\sigma^t}(z) \]
for $z\in \mc Z$.
Then we remark that during \ac{MCCFR} it suffices to keep track of $\textnormal{crp}_i^t(ha)$ where $a\in \mc A(h)$ and $h$ is in the tree built by \ac{MCCFR}, and show how these values can be calculated similarly to the lazy-weighted averaging of~\cite{lanctot_thesis}.
Lastly, we note that a sampled variant can be used in order to get an unbiased estimate of $u_i^{\bar \sigma^T}$.

Recall the standard fact that the average strategy satisfies\begin{equation}\label{eq:avg_util}
u^{\bar \sigma^T_1, \nu_2}_i (h)
= \frac{\sum_t \pi^{\sigma^t_1}_1(h) u^{\sigma^t_1, \nu_2}_i (h)}{\sum_t \pi^{\sigma^t_1}_1(h)}
\end{equation}
for every $h\in \mc H$, and has the analogous property for $\bar \sigma^T_2$.
Indeed, this follows from the formula
\begin{equation}
\pi_i^{\bar \sigma^T} (h)
= \frac{1}{T} \sum_t \pi^{\sigma^t}_i(h) ,
\end{equation}
which can be proven by induction over the length of $h$ using \eqref{eq:avg_strat}.

\begin{lemma}\label{lem:u_of_avg}
For any $h\in \mc H$, $i$ and $\sigma^1,\dots, \sigma^T$, we have $u_i^{\bar \sigma^T}(h) = $
\[
\frac{\sum_t \sum_{z\sqsupset h}
	\left( \pi_1^{\sigma^t}\!\!(z)\textnormal{crp}_2^t(z) + \textnormal{crp}_1^t(z) \pi_2^{\sigma^t}\!\!(z) - \pi^{\sigma^t}_{1,2}(z) \right)
	\pi_c(z|h) u_i(z)	
}
{ \textnormal{crp}_1^T(h) \textnormal{crp}_2^T(h) }
. \]
\end{lemma}

\begin{proof}
For $h\in \mc H$, we can rewrite $u_i^{\bar \sigma^T}(h)$ as\begin{align*}
u_i^{\bar \sigma^T}(h)
& =  u_i^{\bar \sigma^T_1, \bar \sigma^T_2}(h)
\overset{\eqref{eq:avg_util}}{=}
 \frac{\sum_t \pi^{\sigma^t_1}_1(h) u^{\sigma^t_1, \bar \sigma^T_2}_i (h)}
	{\textnormal{crp}_1^T(h)} \\
& \overset{\eqref{eq:avg_util}}{=} \frac{
		\sum_t \pi^{\sigma^t_1}_1(h)
		\frac{\sum_s \pi^{\sigma^s_2}_2(h) u^{\sigma^t_1, \sigma^s_2}_i (h)}
		{\textnormal{crp}_2^T(h)}
	}
	{\textnormal{crp}_1^T(h)} \\
& = \frac{ 
		\sum_t \sum_s \pi^{\sigma^t_1}_1(h) \ \pi^{\sigma^s_2}_2(h)
		\ u^{\sigma^t_1, \sigma^s_2}_i (h)
	}
	{\textnormal{crp}_1^T(h)\textnormal{crp}_2^T(h)} = \frac{N}{D}.
\end{align*}
Using the definition of expected utility, we can rewrite the numerator $N$ as
\begin{align*}
N & = \sum_{s,t} \pi^{\sigma^t_1}_1(h) \ \pi^{\sigma^s_2}_2(h)
	\sum_{z \sqsupset h} \pi^{\sigma^t}_1(z|h) \ \pi^{\sigma^s}_2(z|h)
	\pi_c(z|h) u_i(z) \\
& = \sum_{s,t} \sum_{z \sqsupset h}
	\pi^{\sigma^t_1}_1(z) \ \pi^{\sigma^s_2}_2(z) \
	\pi_c(z|h) u_i(z) \\
& = \sum_{z \sqsupset h} \left( \pi_c(z|h) u_i(z)
	\sum_{s,t} 	\pi^{\sigma^t_1}_1(z) \ \pi^{\sigma^s_2}_2(z) \right) .
\end{align*}

The double sum over $s$ and $t$ can be rewritten using the formula
\[ \sum_t \sum_s  x_t y_s = \sum_t \left[ x_t(y_1+\dots+y_t) + (x_1+\dots+x_t)y_t - x_t y_t \right] , \]
which yields $\sum_{s,t} \pi^{\sigma^t}_1(z) \ \pi^{\sigma^s}_2(z)=$
\begin{align*}
 & = &  \sum_t \Big[ & \pi^{\sigma^t}_1\!(z) \left(\pi^{\sigma^1}_2\!(z)+\dots+\pi^{\sigma^t}_2\!(z) \right) \ + \\
&& & + \ \left( \pi^{\sigma^1}_1\!(z)+\dots+\pi^{\sigma^t}_1\!(z) \right)\pi^{\sigma^t}_2\!(z) - \pi^{\sigma^t}_1\!(z)\pi^{\sigma^t}_2\!(z) \Big] \\
& = & \sum_t \Big[ & \pi^{\sigma^t}_1\!(z) \ \textnormal{crp}^t_2(z) + 
	\textnormal{crp}^t_1 (z) \ \pi^{\sigma^t}_2\!(z) -
	\pi^{\sigma^t}_1\!(z)\pi^{\sigma^t}_2\!(z) \Big] .
 \end{align*}
Substituting this into the formula for $N$ and $\frac{N}{D}$ concludes the proof.
\end{proof}

\subsection{Computing Cumulative Reach Probabilities}
While it is intractable to store $\textnormal{crp}_i^t(z)$ in memory for every $z\in \mc Z$, we \emph{can} store the cumulative reach probabilities for nodes in the tree $\mc T_t$ built by \ac{MCCFR} at time $t$.
We can translate these into $\textnormal{crp}_i^t(z)$ with the help of the uniformly random strategy $\textnormal{rnd}$:

\begin{lemma}
Let $z\in \mc Z$ be s.t. $z \sqsupset ha$, where $h$ is a leaf of $\mc T_t$ and $a\in \mc A(h)$. Then we have $\textnormal{crp}_i^t(z) = \textnormal{crp}_i^t(ha) \pi^\textnormal{rnd}_i(z|ha)$.
\end{lemma}

\begin{proof}
This immediately follows from the fact that for any $g \notin \mc T_t$, $\sigma^s(g) = \textnormal{rnd}(g)$ for every $s=1,2,\dots,t$.
\end{proof}

To keep track of $\textnormal{crp}_i^t(ha)$ for $h\in \mc T_t$, we add to it a variable $crp_i(h)$ and auxiliary variables $w_i(ha)$, $a\in \mc A(h)$, which measure the increase in cumulative reach probability since the previous visit of $ha$.
All these variables are initially set to $0$ except for $w_i(\emptyset)$ which is always assumed to be equal to 1.
Whenever \ac{MCCFR} visits some $h\in \mc T_t$,  is visited, $\textnormal{crp}_i(h)$ is increased by $w_i(h)$ (stored in $h$'s parent), each $w_i(ha)$ is increased by $w_i(h) \pi_i^{\sigma^t}(ha|h)$ and $w_i(h)$ (in the parent) is set to $0$.
This ensures that whenever a value $w_i(ha)$ gets updated without being reset, it contains the value $\textnormal{crp}_i^t(ha) - \textnormal{crp}_i^{t_{ha}}(ha)$, where $t_{ha}$ is the previous time when $ha$ got visited. As a consequence, the variables $\textnormal{crp}_i(ha)$ that do get updated are equal to $\textnormal{crp}^t_i(ha)$.
Note that this method is very similar to the lazy-weighted averaging of~\cite{lanctot_thesis}.

Finally, we observe that the formula from Lemma~\ref{lem:u_of_avg} can be used for on-the-fly calculation of an unbiased estimate of $u_i^{\bar \sigma^T}(h)$.
Indeed, it suffices to replace the sum over $z$ by its sampled variant $\hat s^t_i(h) :=$
\begin{equation}\label{eq:sampled_u_of_avg}
\frac{1}{q^t(z)} \left(
\pi_1^{\sigma^t}\!\!(z)\textnormal{crp}_2^t(z)
+ \textnormal{crp}_1^t(z) \pi_2^{\sigma^t}\!\!(z)
- \pi^{\sigma^t}_{1,2}(z)
\right) \pi_c(z|h) u_i(z)
,
\end{equation}
where $z$ is the terminal state sampled at time $t$ and $q^t(z)$ is the probability that it got sampled with $z$.
We keep track of the cumulative sum $\sum_t \hat s_i^t(h)$ and, once we reach iteration $T$, we do one last update of $h$ and set
\begin{align*}
\tilde u_i(h) := \frac{ \sum_t \hat s^t_i(h)}{ \textnormal{crp}_1^T(h) \textnormal{crp}_2^T(h) }  & \ \ \ \ \textnormal{ and } & 
\tilde v_i(h) := \pi_{-i}^{\bar \sigma^T}(h) \tilde u_i(h) .
\end{align*}
By \eqref{eq:sampled_u_of_avg} and Lemma~\ref{lem:u_of_avg}, we have $\mathbf E \tilde u_i(h) = u_i^{\bar \sigma^T}(h)$ and thus $\mathbf E \tilde v_i(h) = v_i^{\bar \sigma^T}(h)$.
Note that $\tilde v_i(h)$ might suffer from a very high variance and devising its low-variance modification (or alternative) would be desirable.

\section{Game Rules}\label{sec:rules}

\textbf{Biased Rock Paper Scissors} \texttt{B-RPS}
is a version of standard game of Rock-Paper-Scissors with modified payoff matrix:
\begin{table}[H]
    \centering
    \begin{tabular}{c|ccc}
        ~ & R & P & S \\
        \hline
        R & 0 & -1 & 100 \\
        P & 1 & 0 & -1 \\
        S & -1 & 1 & 0
    \end{tabular}
\end{table}
This variant gives the first player advantage and breaks the game action symmetry.

\textbf{Phantom Tic-Tac-Toe} \texttt{PTTT}
Phantom Tic-Tac-Toe is a partially observable variant of Tic-Tac-Toe. It is played by two players on 3x3 board and in every turn one player tries to mark one cell. The goal is the same as in perfect-information Tic-Tac-Toe, which is to place three consecutive marks in a horizontal, vertical, or diagonal row.

Player can see only his own marked cells, or the marked cells of the opponent if they have been revealed to him by his attempts to place the mark in an occupied cell.

If the player is successful in marking the selected cell, the opponent takes an action in the next round. Otherwise, the player has to choose a cell again, until he makes a successful move.

The opponent receives no information about the player's attempts at moves.

\textbf{Imperfect Information Goofspiel} In \texttt{II-GS(N)}, each player is given a private hand of bid cards with values $0$ to $N-1$. A different deck of $N$ point cards is placed face up in a stack. On their turn, each player bids for the top point card by secretly choosing a single card in their hand. The highest bidder gets the point card and adds the point total to their score, discarding the points in the case of a tie. This is repeated $N$ times and the player with the highest score wins.

In \textit{II-Goofspiel}, the players only discover who won or lost a bid, but not the bid cards played. Also, we assume the point-stack is strictly increasing: $0, 1, \dots N-1$. This way the game does not have chance nodes, all actions are private and information sets have various sizes.

\textbf{Liar's Dice} \texttt{LD(D1,D2,F)}, also known as Dudo, Perudo, and Bluff is a dice-bidding game. Each die has faces $1$ to $F-1$ and a star $\star$. Each player $i$ rolls $D_i$ of these dice without showing them to their opponent. Each round, players alternate by bidding on the outcome of all dice in play until one player "calls liar'', i.e. claims that their opponent's latest bid does not hold. If the bid holds, the calling player loses; otherwise, she wins. A bid consists of a quantity of dice and a face value. A face of $\star$ is considered wild and counts as matching any other face. To bid, the player must increase either the quantity or face value of the current bid (or both).

All actions in this game are public. The only hidden information is caused by chance at the beginning of the game. Therefore, the size of all information sets is identical.

\textbf{Generic Poker} \texttt{GP(T, C, R, B)} is a simplified poker game inspired by Leduc Hold'em. First, both players are required to put one chip in the pot. Next, chance deals a single private card to each player, and the betting round begins. A player can either \textit{fold} (the opponent wins the pot), \textit{check} (let the opponent make the next move), \textit{bet} (add some amount of chips, as first in the round), \textit{call} (add the amount of chips equal to the last bet of the opponent into the pot), or \textit{raise} (match and increase the bet of the opponent).

If no further raise is made by any of the players, the betting round ends, chance deals one public card on the table, and a second betting round with the same rules begins. After the second betting round ends, the outcome of the game is determined - a player wins if: (1) her private card matches the table card and the opponent's card does not match, or (2) none of the players' cards matches the table card and her private card is higher than the private card of the opponent. If no player wins, the game is a draw and the pot is split.

The parameters of the game are the number of types of the cards $T$, the number of cards of each type $C$, the maximum length of sequence of raises in a betting round $R$, and the number of different sizes of bets $B$ (i.e., amount of chips added to the pot) for bet/raise actions.

This game is similar to Liar's Dice in having only public actions. However, it includes additional chance nodes later in the game, which reveal part of the information not available before. Moreover, it has integer results and not just win/draw/loss.

No Limit Leduc Hold'em poker with maximum pot size of $N$ and integer bets is $GP(3,2,N,N)$.

\begin{table}[t]
    \centering
    \begin{small}
        \begin{tabular}{rr}
            Game & $|\mc H|$ \\
            \hline
            IIGS(5) & 41331  \\
            IIGS(13) & $\approx 4\cdot 10^{19}$ \\
            LD(1,1,6) & 147456  \\
            LD(2.2,6) & $\approx 2\cdot 10^{10}$  \\
            GP(3,3,2,2) & 23760  \\
            GP(4,6,4,4) & $\approx 8 \cdot 10^{8}$ \\
            PTTT & $\approx 10^{10}$ \\
        \end{tabular}
    \end{small}
    \caption{Sizes of the evaluated games.}\label{tab:sizes}
\end{table}

\section{Extended results}

\begin{table*}[p]
    \begin{tabular}{ccc}
        \texttt{II-GS(5)} & \texttt{LD(1,1,6)} & \texttt{GP(3,3,2,2)} \\
        \includegraphics[width=0.27\linewidth]{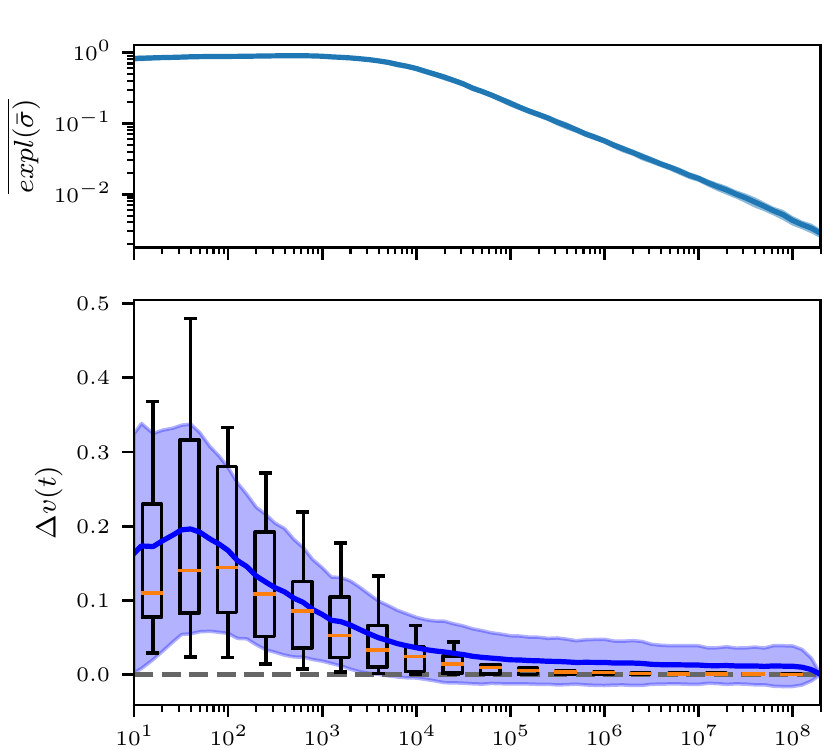} &
        \includegraphics[width=0.27\linewidth]{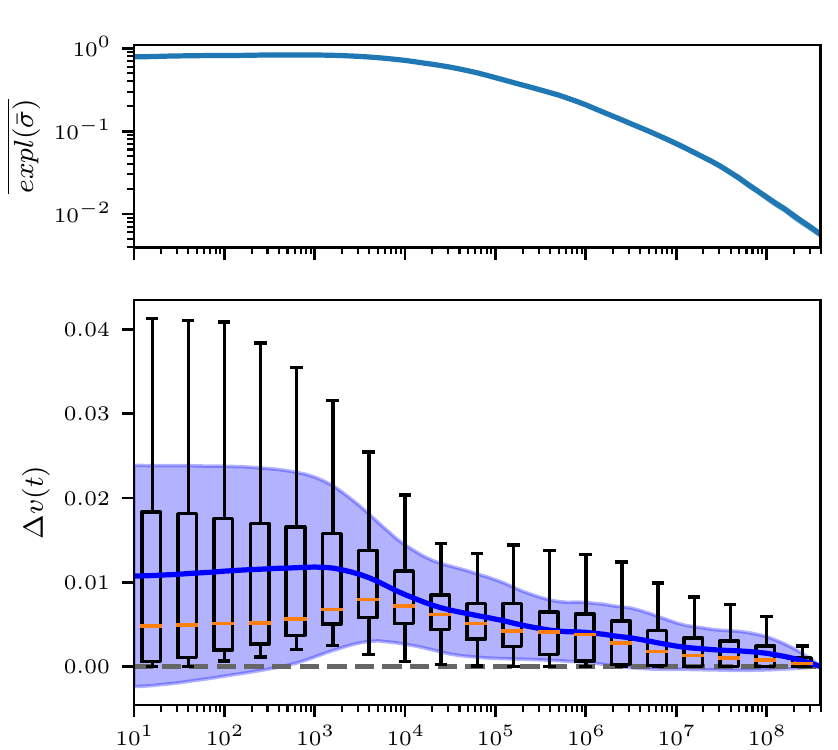} &
        \includegraphics[width=0.27\linewidth]{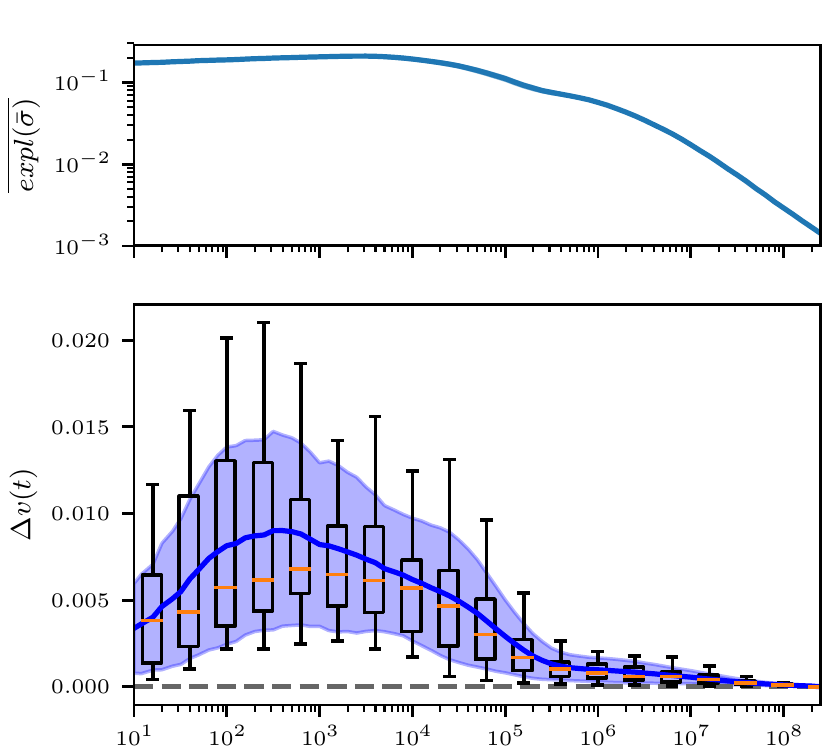} \\
        \texttt{II-GS(13)} & \texttt{LD(2,2,6)} & \texttt{GP(4,6,4,4)} \\
        \includegraphics[width=0.27\linewidth]{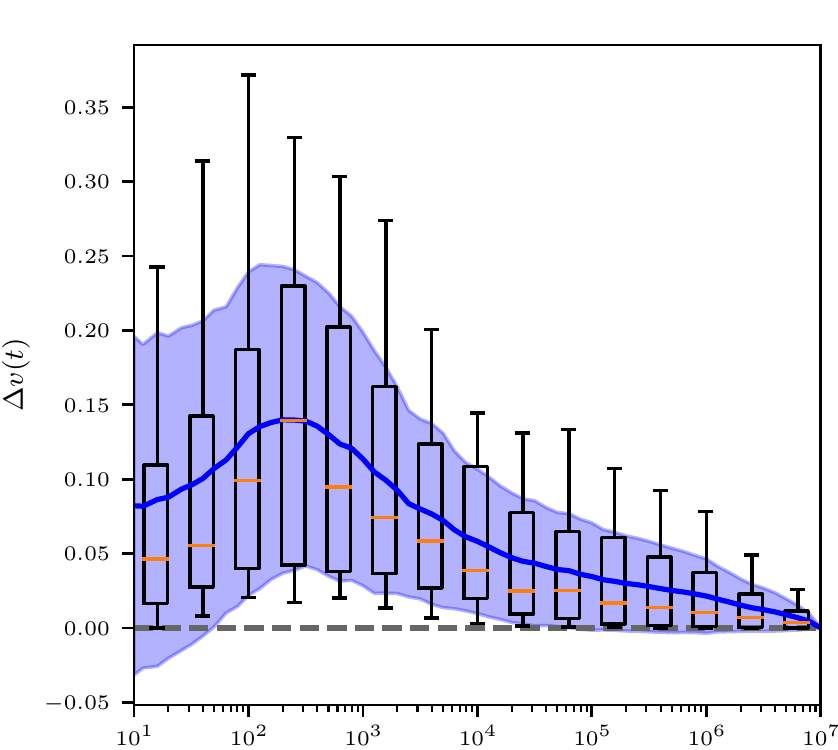} &
        \includegraphics[width=0.27\linewidth]{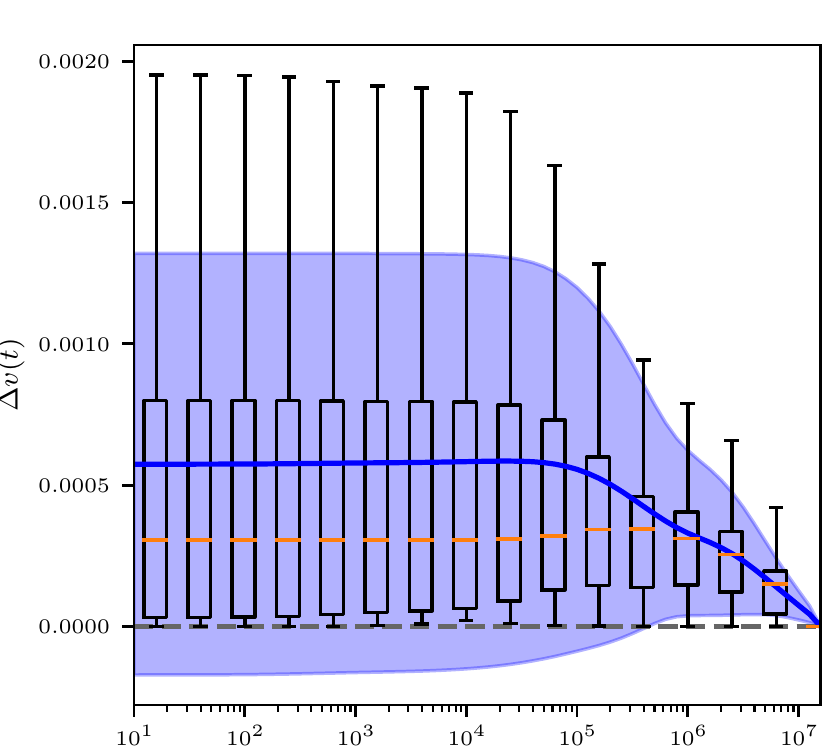} &
        \includegraphics[width=0.27\linewidth]{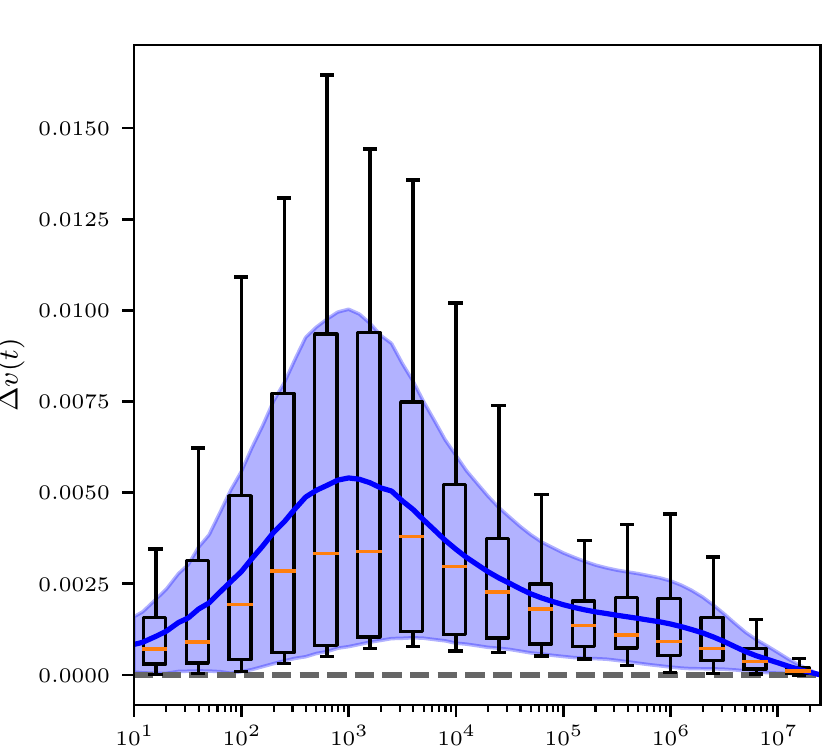} \\
    \end{tabular}
    \caption{Comparison of counterfactual values for different domains by tracking how absolute differences $\Delta_t(J) = | \tilde v_2^t(J) - \tilde v_2^T(J) | $ change over time.}
\end{table*}

\begin{table*}[p]
    \begin{tabular}{ccc}
        \texttt{II-GS(5)} & \texttt{LD(1,1,6)} & \texttt{GP(3,3,2,2)} \\
        \includegraphics[width=0.27\linewidth]{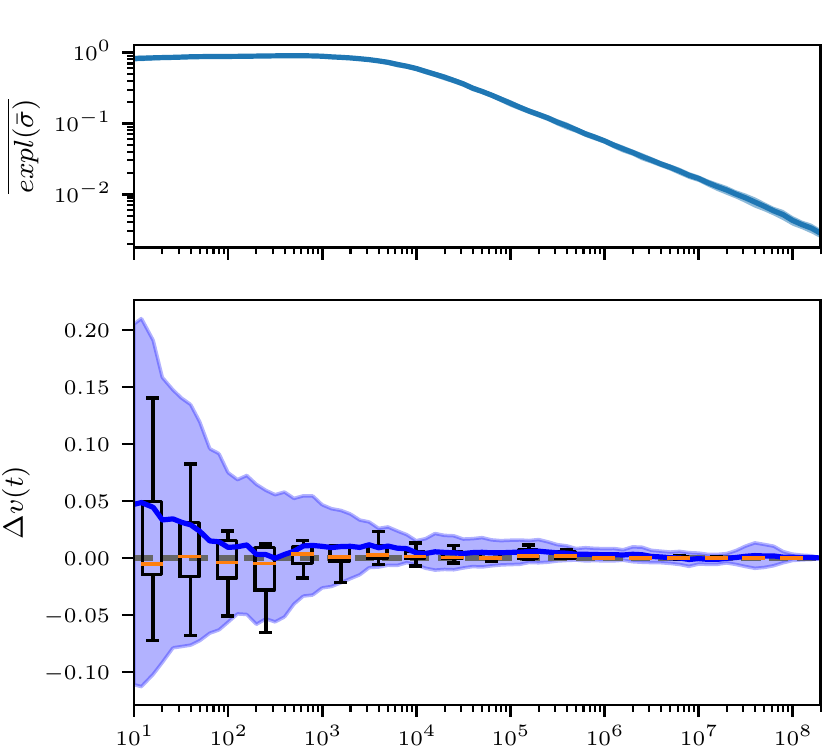} &
        \includegraphics[width=0.27\linewidth]{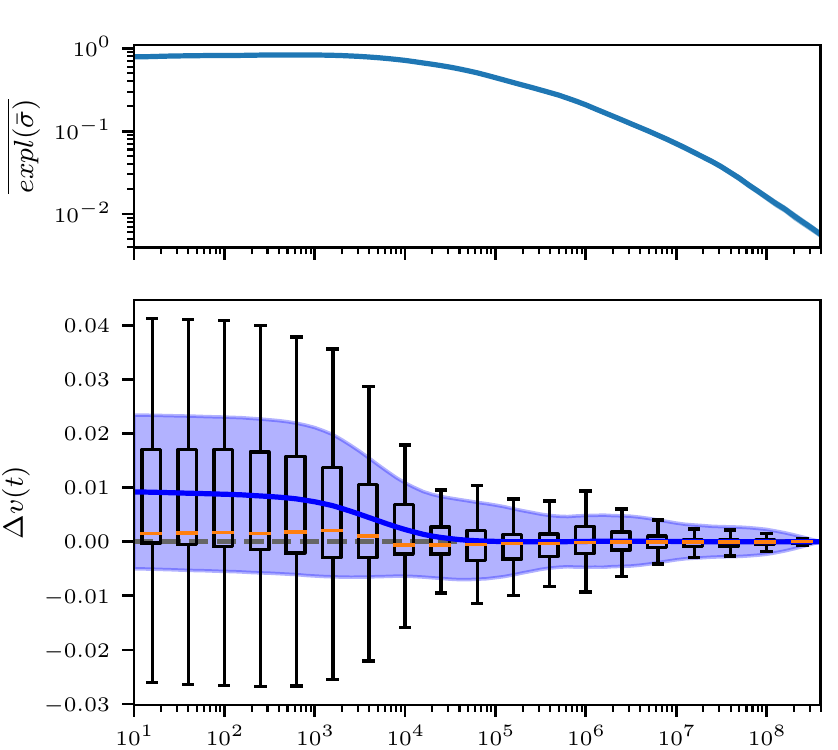} &
        \includegraphics[width=0.27\linewidth]{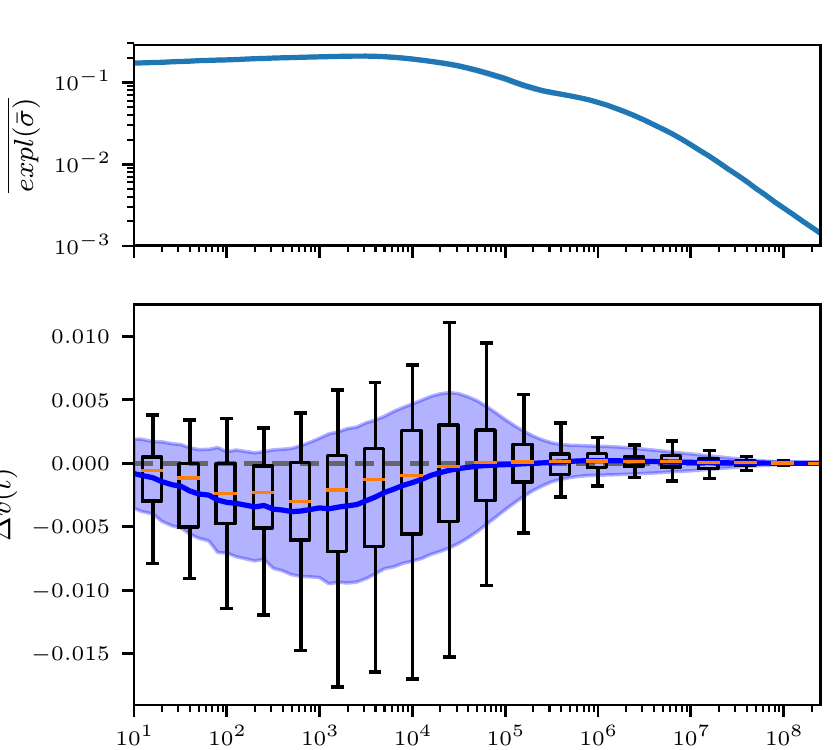}  \\
        \texttt{II-GS(13)} & \texttt{LD(2,2,6)} & \texttt{GP(4,6,4,4)} \\
        \includegraphics[width=0.27\linewidth]{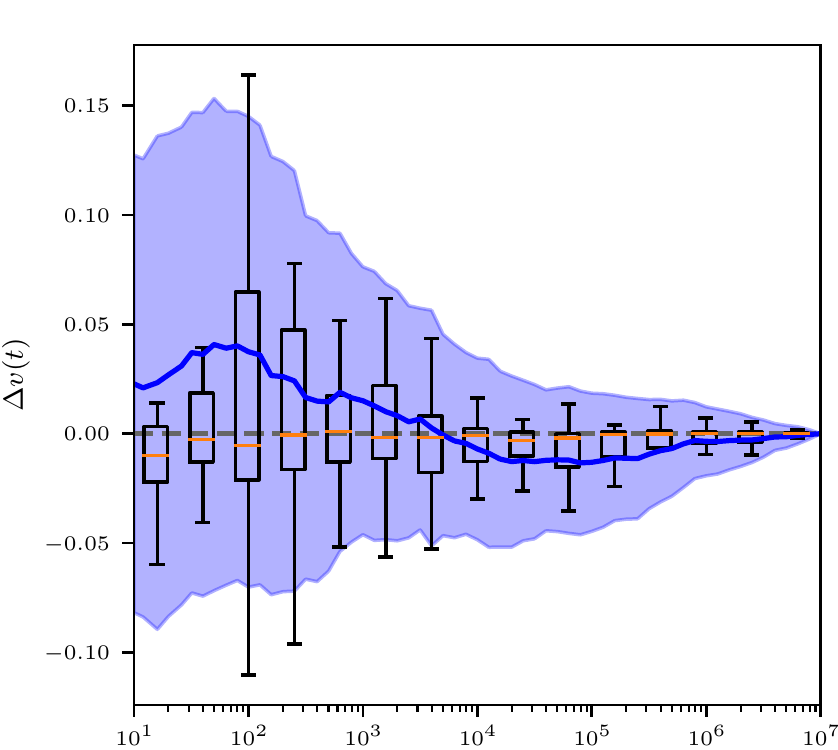} &
        \includegraphics[width=0.27\linewidth]{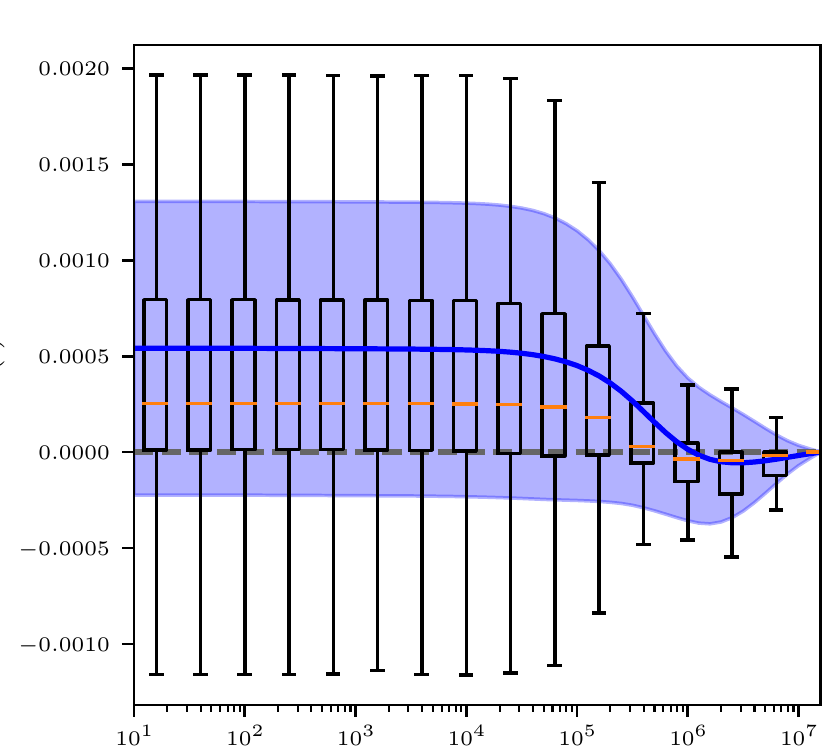} &
        \includegraphics[width=0.27\linewidth]{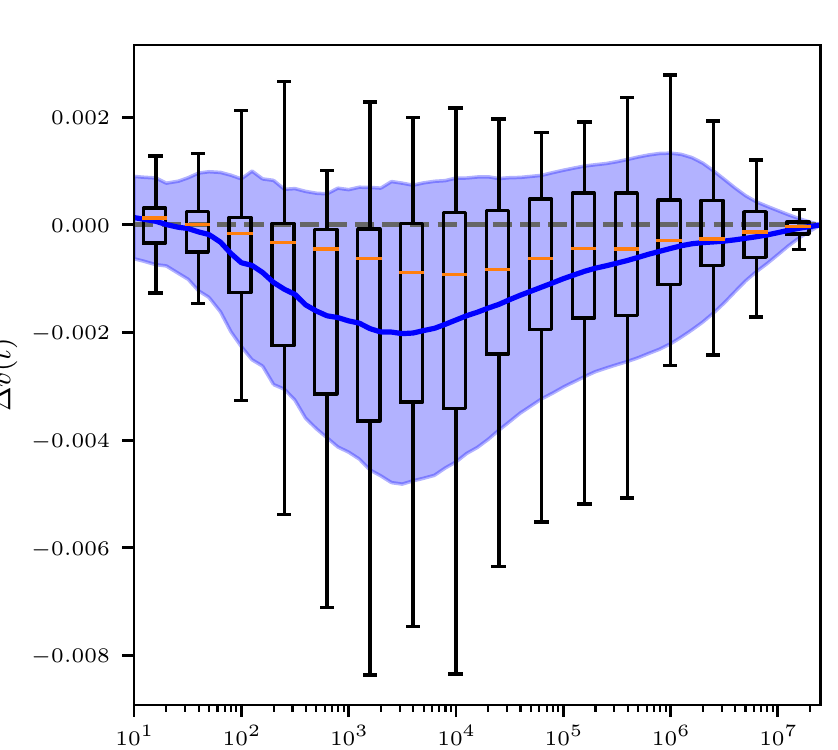}  \\
    \end{tabular}
    \caption{Comparison of counterfactual values for different domains by tracking how relative differences $\Delta_t(J) = \tilde v_2^t(J) - \tilde v_2^T(J) $ change over time.}
\end{table*}

\begin{table*}[p]
    \begin{tabular}{ccc}
        \texttt{II-GS(5)} & \texttt{LD(1,1,6)} & \texttt{GP(3,3,2,2)} \\
        \includegraphics[width=0.27\linewidth]{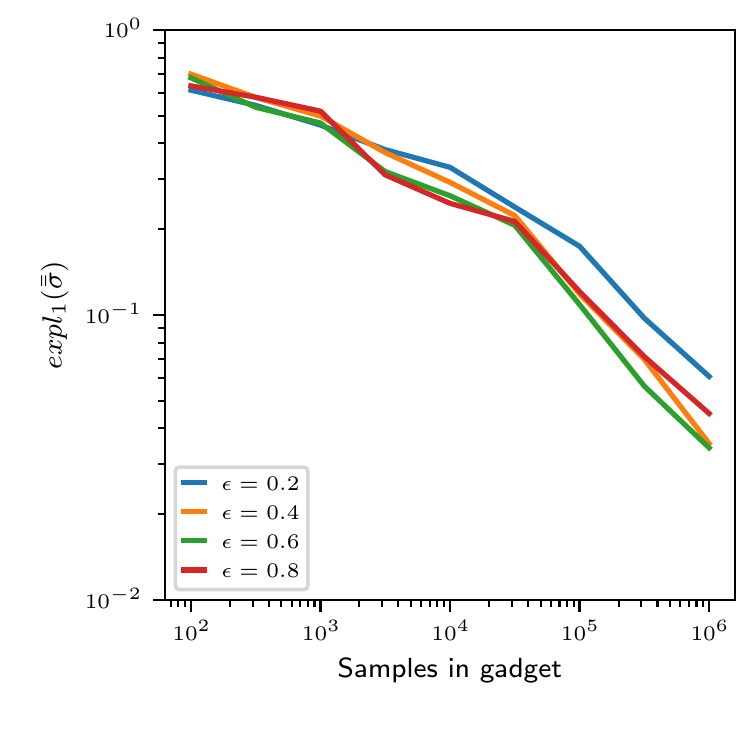} &
        \includegraphics[width=0.27\linewidth]{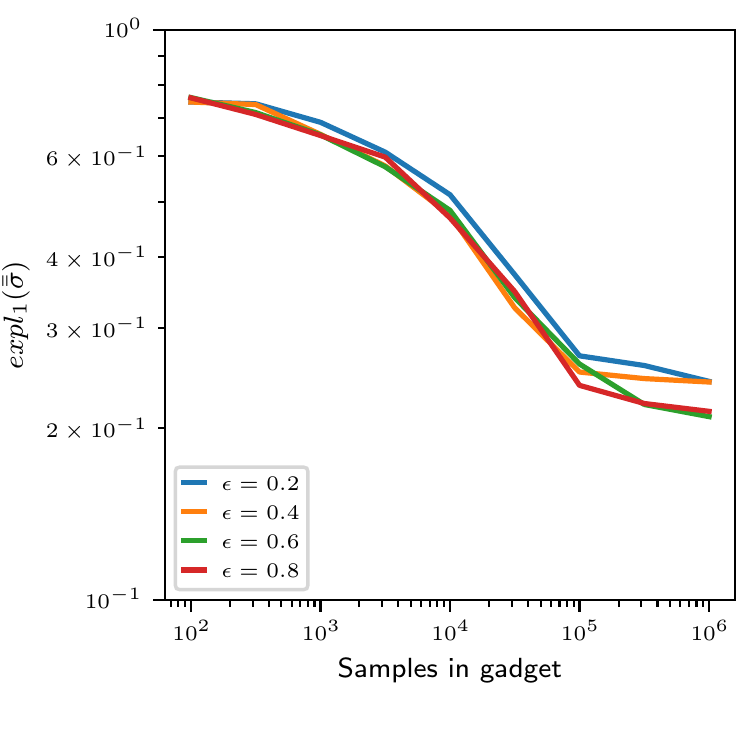} &
        \includegraphics[width=0.27\linewidth]{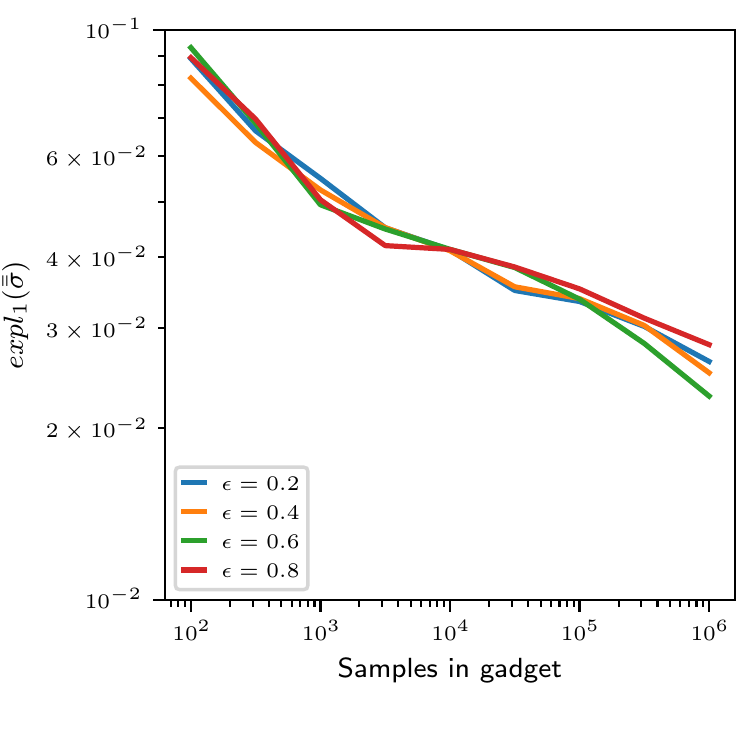}  \\
        \includegraphics[width=0.27\linewidth]{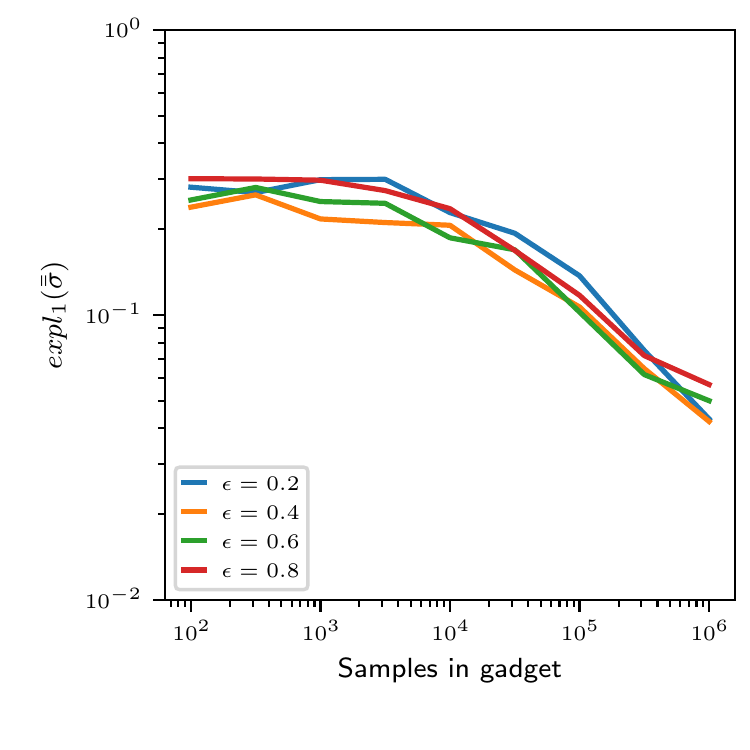} &
        \includegraphics[width=0.27\linewidth]{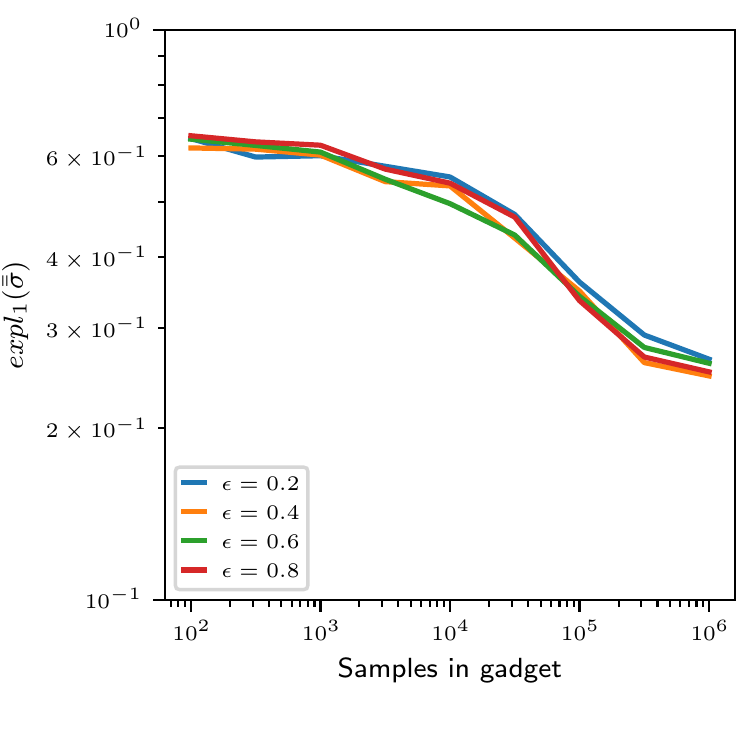} &
        \includegraphics[width=0.27\linewidth]{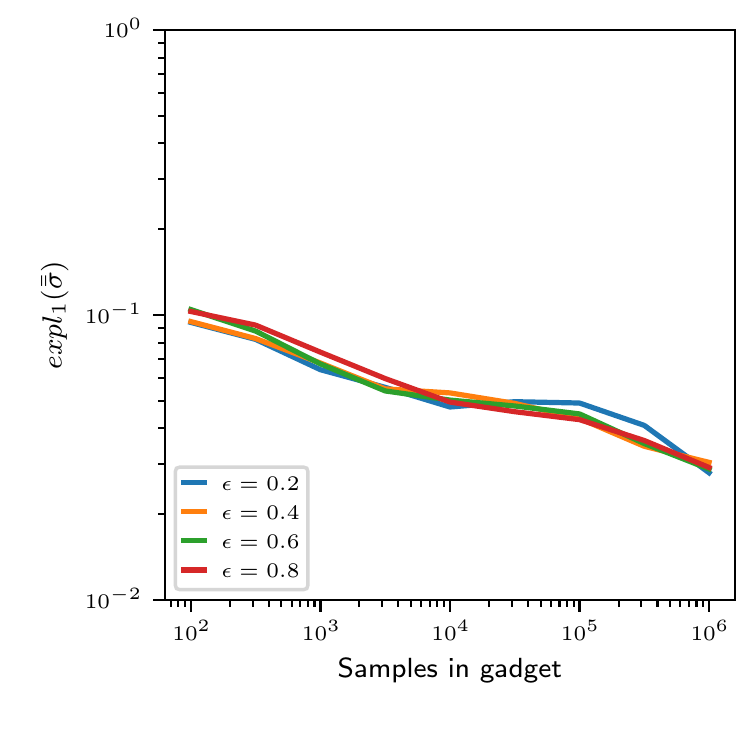}
    \end{tabular}
    \caption{Sensitivity to exploration parameter. Top row is "reset" variant, bottom row is "keep" variant. }
\end{table*}

\end{document}